\pdfoutput=1
\documentclass[pra,amsmath,amssymb,amsfonts,nofootinbib]{revtex4}
\usepackage{times}
\usepackage{amsmath, amsthm, amssymb,amsfonts, bbm}

\usepackage{color}
\usepackage{graphicx}
\usepackage{tikz}
\usetikzlibrary{positioning, arrows.meta, calc}

\usepackage{enumerate}

\newcommand{\bR}{\mathbb{R}}
\newcommand{\bN}{\mathbb{N}}
\newcommand{\bC}{\mathbb{C}}
\newcommand{\bZ}{\mathbb{Z}}

\newcommand{\bE}{\mathbb{E}}

\newcommand{\cB}{\mathcal{B}}
\newcommand{\cH}{\mathcal{H}}

\newcommand{\cG}{\mathcal{G}}

\newcommand{\PEPS}{{\rm PEPS}}

\usepackage{physics}
\newcommand{\proj}{\dyad}

\newcommand{\1}{\mathbbm{1}}
\newcommand{\id}{{\rm id}}

\newcommand{\pathnorm}[1]{{\left\vert\kern-0.25ex\left\vert\kern-0.25ex\left\vert #1
        \right\vert\kern-0.25ex\right\vert\kern-0.25ex\right\vert}}

\DeclareMathOperator{\OExp}{Exp}
\DeclareMathOperator{\diam}{diam}
\DeclareMathOperator{\dist}{dist}

\newcommand{\be}{\begin{equation}}
\newcommand{\ee}{\end{equation}}
\newcommand{\bea}{\begin{eqnarray}}
\newcommand{\eea}{\end{eqnarray}}
\newcommand{\bes}{\begin{equation*}}
\newcommand{\ees}{\end{equation*}}
\newcommand{\beas}{\begin{eqnarray*}}
\newcommand{\eeas}{\end{eqnarray*}}

\theoremstyle{plain}
\newtheorem{thm}{Theorem}
\newtheorem{corollary}[thm]{Corollary}
\newtheorem{conjecture}[thm]{Conjecture}
\newtheorem{lemma}[thm]{Lemma}

\newtheorem{prop}[thm]{Proposition}
\newtheorem{res}{Main Result}
\newtheorem*{thm*}{Theorem}
\newtheorem*{lem*}{Lemma}

\theoremstyle{definition}
\newtheorem{definition}[thm]{Definition}

\theoremstyle{remark}
\newtheorem{rem}[thm]{Remark}

\begin{document}

\title{\sc Locality at the boundary implies gap in the bulk for 2D PEPS}
\author{Michael J. Kastoryano$^{1,2}$, Angelo Lucia$^{1,3,4}$, David Perez-Garcia$^{5,6}$}
\address{
  $^1$ NBIA, Niels Bohr Institute, University of Copenhagen, Denmark \\
  $^2$ Institute for Theoretical Physics, University of Cologne, Germany \\
  $^3$ QMATH, Department of Mathematical Sciences, University of Copenhagen, Denmark  \\
  $^4$ Walter Burke Institute for Theoretical Physics and Institute for Quantum Information \& Matter,
  California Institute of Technology, Pasadena, CA 91125, US \\
  $^5$ Dpto. An\'alisis Matem\'atico, Universidad Complutense de Madrid, 28040 Madrid, Spain \\
  $^65$ Instituto de Ciencias Matem\'aticas, 28049 Madrid, Spain
}

\date{\today}

\begin{abstract}
  Proving that the parent Hamiltonian of a Projected Entangled Pair State (PEPS) is gapped remains
  an important open problem. We take a step forward in solving this problem by showing two results:
  first, we identify an approximate factorization condition on the boundary state of rectangular
  subregions that is sufficient to prove that the parent Hamiltonian of the bulk 2D PEPS has a
  constant gap in the thermodynamic limit; second, we then show that Gibbs state of a local,
  finite-range Hamiltonian satisfy such condition. The proof applies to the case of injective and
  MPO-injective PEPS, employs the martingale method of nearly commuting projectors, and exploits a
  result of Araki \cite{Araki} on the robustness of one dimensional Gibbs states. Our result
  provides one of the first rigorous connections between boundary theories and dynamical properties
  in an interacting many body system.
 \end{abstract}
 \maketitle

\tableofcontents

\section{Introduction}

Quantum Information Theory (QIT) and the Theory of Quantum Many Body (QMB) systems are inextricably
connected. On the one hand, QMB systems and their properties, such as entanglement or topological
order, play a crucial role in the design of quantum computers and quantum simulators. On the other
hand, the use of QIT tools and ideas have shed new light on the structure, properties and mechanisms
present in QMB systems.

An important part of these developments have come through the so called Tensor Networks States
(TNS), variational families of states which mimic the entanglement structure present in ground
states and thermal states of QMB systems. Indeed, as proven in a series of papers, Projected
Entangled Pair States (PEPS), a particularly relevant family of TNS, approximate well ground and
thermal states of local Hamiltonians \cite{hastings2006solving, hastings2007area, arad2013area,
  brandao2015exponential, molnar2015approximating}. Based on this, they have been used to design
better algorithms to simulate QMB systems (see e.g. \cite{orus2014practical} for a recent review on
this). A remarkable example in this direction is \cite{arad2016rigorous, roberts2017rigorous}, where
a new algorithm to approximate ground states of gapped spin chains is presented and shown
competitive compared to DMRG, while providing a guarantee of fast convergence.

Being arbitrarily close to any ground or thermal state of a QMB system, PEPS have been used also as
a new analytical tool, based on QIT techniques, to give a rigorous mathematical treatment of some of
the challenges posed by QMB systems. In this direction one can highlight for instance the
classification given in \cite{pollmann2010, schuch2011class, chen2011classification,
  fidkowski2011topological} of (symmetry-protected) phases in 1D systems, or the microscopic
explanation of topologically ordered systems of \cite{schuch2010, buerschaper2014, MPOinjective,
  cirac2017, bultinck2017}. Since a PEPS is always the ground state of an associated short range
Hamiltonian (called the parent Hamiltonian), this led for instance to the first local Hamiltonian
having the Resonating Valence Bond (RVB) State as unique ground state in the kagome lattice (up to
topological degeneracy) \cite{schuch2012resonating}, a question that can be traced back to the
seminal work of P. W. Anderson in the 70s \cite{anderson1973resonating}. Recently, TNS started to
play also a key role in providing a rigorous framework to address some of the main open problems in
high energy physics, such as those coming from Maldacena's AdS/CFT holographic correspondence
\cite{Swingle2012, pastawski2015holographic, hayden2016holographic,kim2017entanglement}.

A key insight in all these analytical developments is the existence of suitable connections between
the bulk of the system and its boundary in every TNS. Such bulk-boundary correspondences have also
played an important role in the study in QMB physics recently thanks to the seminal paper of Li and
Haldane \cite{li2008ES}, where they observed that the spectrum of the reduced density matrix of a
fractional quantum Hall ground state closely resembled a conformal field theory at the boundary
\cite{li2008ES}. This observation led the authors to conjecture that the entanglement spectrum
contains information about the counting of edge states.

Following the path initiated by Li and Haldane, the bulk-boundary correspondence in PEPS was made
explicit in \cite{ciracES} via an isometry and the definition of the so called boundary
Hamiltonian. This is an auxiliary 1D interaction whose thermal state (called boundary state) has two
key features: (1) its spectrum can be associated to the spectrum of the reduced density matrix of
the given PEPS in the bulk and (2) its (boundary) correlations can be associated to the correlations present in the bulk.

Since thermal states of local, finite-range Hamiltonians have exponentially decaying correlations (a
result due to Araki in \cite{Araki}), the following conjecture was stated in \cite{ciracES}:
\begin{conjecture}\label{conjecture-bulk-boundary}
The bulk parent Hamiltonian of the PEPS is gapped {\rm if and only if} the boundary Hamiltonian is short range.
\end{conjecture}
This conjecture was verified numerically in \cite{ciracES} for Ising PEPS \cite{verstraete2006criticality},
one of the few models for which the bulk spectral gap is analytically proved.  More generally,
Conjecture \ref{conjecture-bulk-boundary} has led to new ways of numerically analyzing spectral gaps
and quantum phase transitions in PEPS-based algorithms by looking at the boundary (see
e.g. \cite{poilblanc2012topological, schuch2013topological, poilblanc2014entanglement,
  rispler2015long, gauthe2017entanglement}).

Conjecture \ref{conjecture-bulk-boundary} also opens a new line of thought in the problem of
characterizing analytically when a given 2D (or higher dimensional) PEPS has a gapped parent
Hamiltonians.  This problem has seen virtually no progress since the original AKLT paper (see
however \cite{nachtergaele1996,knabe1988}), even though it was stated as an important open problem
in Ref. \cite{AKLT}. Its difficulty can be easily guessed from the recent result that proving
spectral gaps in 2D short range Hamiltonians is undecidable \cite{cubitt2015undecidability}. Despite
this, the problem is undoubtedly important. Apart from the potential solution of the long-standing
spectral gap problem for the 2D AKLT Hamiltonian, a sufficiently sharp condition for the existence
of a spectral gap in PEPS Hamiltonians seems crucial to attack two other central open problems in
QMB systems.

The first one is a mathematically complete classification of gapped quantum phases in 2D, at least
for the cases which can be characterized well within the PEPS framework (groundstates of gapped
non-chiral Hamiltonians are believed to have an efficient description in terms of PEPS, while the
situation for gapped chiral models is less clear \cite{Yang_2015,Dubail_2015,
  Poilblanc_2015,Wahl_2013}). Let us recall that two Hamiltonians are said to be in the same phase
if they can be deformed into each other without crossing a phase transition point (i.e. closing the
gap).  Recent work on Renormalization Fixed Points (RFP) \cite{cirac2017} allows to conclude then
(up to subtle but important technical details) that essentially only Levin-Wen string-net models
\cite{levin2005} appear as 2D RFP. The main remaining step is then to connect with a gapped path of
(parent) Hamiltonians each PEPS with one of the RFP models. Clearly this can only be done if one
finds a suitable criterion to guarantee the presence of such gap.

The second potential application of a sharp gap criterion is a mathematical proof of the existence a
gapped $SU(2)$-invariant topological spin liquid phase. Topological spin liquids constitute a
long-sought new quantum phase that can be traced back to the revolutionary work of Anderson
\cite{anderson1973resonating}, and for which experimental evidence has been observed, even in
naturally occurring materials such as herbertsmithite, although in a gapless
form \cite{Norman_2016}. They are characterized by the absence of order even at zero temperature, and
in the case of gapped spin liquids they are expected to show topological order and quasi-particle
excitations \cite{Savary2016}. A long standing conjecture regards the existence of a gapped spin
liquid which does not break a local $SU(2)$ symmetry, for which the most promising candidate is the
RVB state. The only remaining step to prove is to show that the parent Hamiltonian for the RVB
constructed in \cite{schuch2012resonating} is in the same phase as the Toric Code.  An interpolating
path of parent Hamiltonians connecting the Toric Code and the RVB is constructed in
\cite{schuch2012resonating} and numerical evidence is provided that no phase transition occur along
the path. But the lack of a suitable gap characterization makes it hard to prove analytically the
existence of a gap throughout the path.

In this paper, we take a first important step in solving the PEPS gap problem following the approach
of Conjecture \ref{conjecture-bulk-boundary}. We identify one condition on the boundary state of a
2D injective PEPS, which we call \emph{approximate factorization}, and show that this condition is
sufficient to prove a spectral gap in the bulk.
\begin{res}\label{main-result-I}
  If for any rectangular subregion of the lattice, the boundary state on the virtual indices of an
  injective 2D PEPS is approximately factorizable then the gap of the bulk parent Hamiltonian is bounded below by a constant, in
  the thermodynamic limit.
\end{res}
The result can be extended to topologically ordered systems described by MPO-injective PEPS in a
natural way. In this case, the boundary Hamiltonian is given by a sum of quasi-local terms plus one
global projector that commutes with all other terms, and specifies the topological sector, and the
approximate factorization condition is required on the non-topological part of the boundary state.

The proof strategy of Main Result~\ref{main-result-I} is based on the so-called martingale method: a general strategy
that relates the gap of a frustration free Hamiltonian to the approximate commutativity of
overlapping projectors. The martingale method is inspired from the classical proofs for showing
rapid mixing of Glauber dynamics of Ising type models at finite temperature
\cite{martinelli2004}. Nachtergaele first adapted this strategy for proving the gap of one
dimensional VBS models \cite{nachtergaele1996}. Here we use a slightly different version of the
proof due to Bertini et al. \cite{bertini2002} which was generalized for frustration free
Hamiltonians in \cite{Angelo} and for commuting quantum Gibbs samplers in
\cite{GibbsSampling}. Subsequently, the martingale condition was shown to be implied by the gap
\cite{Angelo,GibbsSampling}. The basic idea underlying the proof is to show that when we double the
system size, the gap remains almost unchanged. Then one can grow the size indefinitely while keeping
the gap constant. In this paper, we will prove the version of the martingale condition
(Eqn.~(\ref{eqn:martingale})) due to Bertini et al. (see Fig.~\ref{fig:ABC}) and refer back to
Ref. \cite{Angelo} for the details of how this condition implies a gap of the parent Hamiltonian in
the thermodynamic limit.

In the setting of Glauber dynamics for classical spin systems, the martingale condition on the
ground state projectors can be rather easily shown to follow from decay of correlations in the bulk
Gibbs state because of the (Dobrushin-Lanford-Ruelle) DLR theory of boundary conditions for
classical Gibbs states \cite{dobrushin1968,lanford1969}. The DLR theory roughly states that the set
of reduced states of all global Gibbs states restricted to a region is equal to the set of all local
Gibbs states on that region with arbitrary boundary conditions. Since this set is convex, one can
further restrict attention to states with pure boundary conditions, which in the classical setting
are product states. Unfortunately, the DLR theory breaks down in the quantum setting both for
proving gaps of quantum Gibbs samplers and for proving gaps of frustration free Hamiltonians
\cite{fannes1995}. Therefore the connection between bulk correlations and the martingale condition
is lost. In short, the dynamic vs. static equivalence that the martingale condition provides in
the classical setting is lost in the quantum setting.
Main Result~\ref{main-result-I} shows that for PEPS this equivalence can be
partially recovered: for injective PEPS, proving the martingale condition
reduces to showing that the boundary states on the rectangular regions of
Fig.~\ref{fig:ABC} can be decomposed in a special way that we call
approximate factorization. Approximate factorization of a state $\rho_{ABC}$ on
a line $ABC$ broken up such that $B$ separates $A$ from $C$, claims that the
state can be written as $\rho_{ABC}\approx \Omega_{AB}\Sigma_{BC}$ for operators
$\Omega_{AB}$ and $\Sigma_{BC}$ with support on $AB$ and $BC$ respectively
(Eqns.~(\ref{eqn:sigma},\ref{QF1},\ref{QF2})). The actual condition of
approximate factorization in Def.~\ref{def:qfact} is more complicated, and
involves several boundary states on rectangular regions as in
Fig.~\ref{fig:bdregions}.

We then show that locality of the interactions in the boundary Hamiltonian is sufficient to prove
the approximate factorization condition.
\begin{res}\label{main-result-2}
  Gibbs (thermal) states of a one dimensional local finite-range Hamiltonian are approximately factorizable.
\end{res}
The proof of Main Result~\ref{main-result-2} is based on a careful use and generalization of the techniques developed
by Araki in \cite{Araki}. The limitation on faster than exponentially decaying interactions arises
from the generalization of these techniques, which we are only able to prove for this case.
By proving a generalization of Main Result~\ref{main-result-2} to the case of exponentially decaying interactions,
which are the ones suggested by the numerical evidence, one would obtain a proof of one half of
Conjecture~\ref{conjecture-bulk-boundary}, namely that a short range boundary Hamiltonian implies a
bulk gap.  Whether Main Result~\ref{main-result-2} admits such generalization is one of the open questions
left in this paper. We will show how one can extend Main Result~\ref{main-result-2} to interactions decaying faster
than any exponential, if one assumes that Araki's theorem holds in this case. In the following we
will call such interactions quasi-local. Improving this result to a plain exponential, if true,
seems a hard task since it is related with the believed but unproven fact that thermal states of
1D short range Hamiltonians with exponentially decaying interactions have exponential decay of
correlations.

The paper is organized as follows.
In Sec.~\ref{sec:prelim} we introduce the notation and overview
the basics of PEPS including the construction of the parent Hamiltonian. We also define boundary
states of PEPS for rectangular regions of the lattice, and we introduce the equivalence between the
martingale condition and the spectral gap of the parent Hamiltonian of the PEPS \cite{Angelo}.
In Sec.~\ref{sec:main} we show that, for injective PEPS, the martingale
condition can be reduced to the approximate factorization property of the
boundary states on rectangular regions. Moreover, we consider the case of a PEPS
defined on a 1D chain, which are known as Matrix Product States (MPS). We show
that the boundary states of injective MPS are always approximately factorizing,
providing an independent proof that the parent Hamiltonian of injective MPS is
gapped \cite{FCS}.
In Sec.~\ref{sec:topo}, we show how our main theorem relating
approximate factorization to the martingale condition can be extended to MPO-injective PEPS in a
natural way.
Sec.~\ref{sec:qfact} is devoted to the Main Result~\ref{main-result-2}, namely that for Gibbs
states of a one dimensional finite-range local Hamiltonian (not to be confused with the parent
Hamiltonian of the PEPS) the approximate factorization property holds. This section requires a
number of tools and results on the analysis of Gibbs states of local Hamiltonians on a line
\cite{Araki}. In order to compare states on different overlapping regions, we need to
introduce the assumptions of \textit{locality} and \textit{homogeneity} on the Hamiltonian.
Finally, we show how an extension of Araki's theorem to the case of quasi-local interaction, which
we conjecture holds but we do not provide a proof of, can be straightforwardly used to generalize our
results to this more general class of interactions.
In the conclusion, we discuss further problems and implications of our theorem, as well as the
relationship between exponential decay of correlations in the bulk and the fact that the boundary
states are quasi-local and quasi-homogeneous Gibbs states.

\section{Preliminaries}\label{sec:prelim}

\subsection{PEPS basics}
We will consider $\Lambda$ to be a finite subset of an infinite graph $(G,E)$, which can be
isometrically embedded in $\bR^2$ (the standard examples being the square lattice $\bZ^2$ or the
honeycomb lattice).  At each $u\in G$ we will associate a finite-dimensional Hilbert space $\cH_d$
of some fixed dimension $d$.  We will denote by $\cH_\Lambda = \bigotimes_{u \in \Lambda} \cH_d$ the
Hilbert space associated to $\Lambda$, and $\cB(\cH_\Lambda)$ is the set of bounded operators on
$\cH_\Lambda$.  If $\{\ket{j_i}\}_{i=1}^d$ is an orthonormal basis for the Hilbert space at site
$j$, any pure state in $\cH_\Lambda$ can be written as
\begin{equation}
  \ket{\psi}=\sum_{j_1,\cdots,j_N=1}^d R_{j_1,\cdots,j_N}\ket{j_1}\otimes\cdots\otimes \ket{j_N},
\end{equation}
where $N = |\Lambda|$ and $R_{j_1,\cdots,j_N}$ is the vector of amplitudes of $\ket{\psi}$, which we
can think of as a tensor in $(\cH_d^*)^{\otimes N}$.

Projected Entangled Pair States (PEPS) are a class of pure states for which it is possible to find a
special description of the tensor $R$. They are constructed by associating to each edge $e\in E$ a
maximally entangled state $\ket{\omega_e}=D^{-1/2}\sum_{j=1}^D\ket{j,j}$, where $D$ is called the
bond dimension, and to each vertex $v\in G$ a linear map $T_v:\cH_D^{\otimes r} \to \cH_d$, where
$r$ is the degree of the vertex $v$ (the number of edges incident to $v$). The Hilbert space
associated with the edges is called the virtual space, whereas the one associated to vertices is
called the physical space. $T_v$ can be considered as a map from the virtual space associated to the
edges $e$ connected to $v$, onto the physical space at $v$:
\begin{equation}
T_v=\sum_{k_v=1}^d\sum_{j_1,\cdots,j_r=1}^D T^{k_v}_{j_1,\cdots,j_r}\ket{k_v}\bra{j_1,\cdots,j_r}.
\end{equation}

\begin{figure}[h]
\centering
  \includegraphics[scale=0.40]{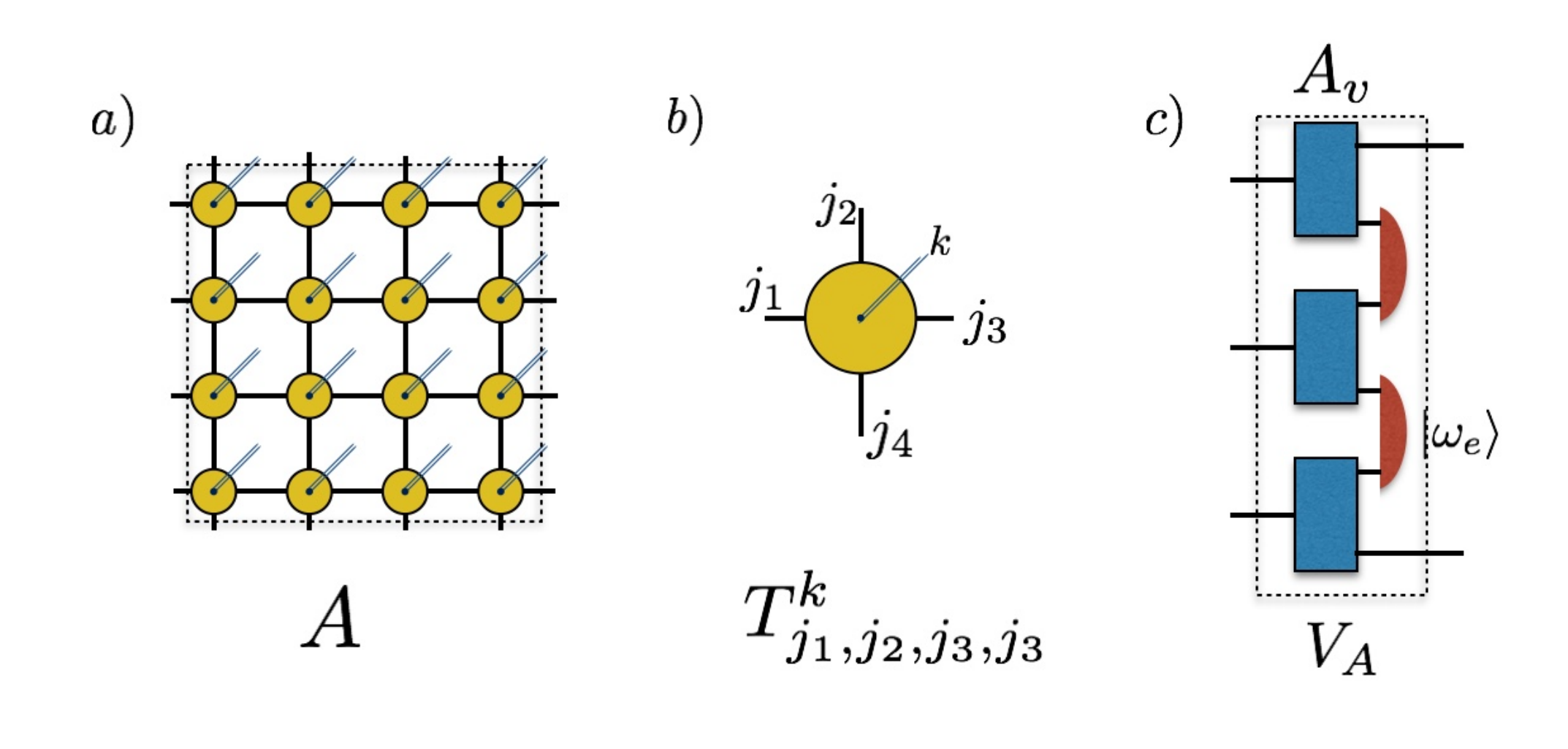}
  \caption{We consider a PEPS state on a square lattice. a) Section of the PEPS on region
    $A\subseteq\Lambda$, b) Graphical representation of the tensors $T^k_{j_1,j_2,j_3,j_4}$, c)
    representation of the operator $V_A$ on a one dimensional lattice. The operator is to be read as
    mapping virtual indices (from the right) to physical indices (to the left).  }
    \label{fig:PEPS1}
\end{figure}
Then $(T^{k_v}_{j_1,\cdots,j_r})$ is a tensor with one physical index and $r$ virtual indices. On a
square lattice, as in Fig.~\ref{fig:PEPS1}, $r=4$.  If $\Lambda$ has no outgoing edges, then we can
define a state in $\cH_\Lambda$ via a tensor contraction: \be \ket{\PEPS_\Lambda}=\bigotimes_{v\in
  \Lambda}T_v\bigotimes_{e\in E_\Lambda}\ket{\omega_e}.\label{eqn:PEPS}\ee If instead there are
edges connecting $\Lambda$ with its complement in $(G,V)$, then we obtain a state in $\cH_\Lambda$
for each choice of ``boundary condition'', in the following sense: denote with $E_{\bar \Lambda}$
the edges which are incident to $\Lambda$, with $E_\Lambda$ the edges with are contained in
$\Lambda$, and with $\partial \Lambda = E_{\bar \Lambda} \setminus E_\Lambda$ the edges that connect
$\Lambda$ with its complement.  Let
$\cH_{\partial \Lambda} = \bigotimes_{e \in \partial \Lambda} \cH_D$ (note that while at each edge
we associated $\ket{\omega_e}\in \cH_D \otimes \cH_D$, we are only including one copy of $\cH_D$ in
$\cH_{\partial \Lambda}$). Then for each vector $\ket{X} \in \cH_{\partial \Lambda}$ we can define a
state
\begin{equation}
\ket{\PEPS_{\Lambda, X}} = \bra{X} \bigotimes_{v \in \Lambda} T_v \bigotimes_{e\in E_{\bar \Lambda}} \ket{\omega_e}.
\end{equation}
This defines a linear map from $\cH_{\partial \Lambda}$ to $\cH_{\Lambda}$, which we will denote
with $V_\Lambda$. It is a mapping from the virtual indices at the boundary of $\Lambda$ to the
physical indices in the bulk of $\Lambda$ (see Fig.~\ref{fig:PEPS1} for an illustration):
\begin{align*}
  V_\Lambda : &\cH_{\partial \Lambda} \to \cH_{\Lambda} \\
  & \ket{X} \mapsto \ket{\PEPS_{\Lambda, X}} .
\end{align*}
A PEPS is said to be \textit{injective on $\Lambda$} \cite{PEPS1} if $V_\Lambda$ is an injective map.
As shown in Ref.~\cite{PEPS1}, if a PEPS is injective on disjoint regions $A$ and $B$, it is also injective on $A\cup B$, so we will simply assume, up to coarse graining of the lattice, that $V_\Lambda$ is injective for every finite $\Lambda$.

Again following Ref.~\cite{PEPS1}, for any injective PEPS, we can define a local Hamiltonian, called the \textit{parent Hamiltonian}, for which the PEPS is the unique groundstate. This is done by considering, for each edge $e = (a,b)$, the orthogonal projector $h_e$ on the orthogonal complement of $\operatorname{Im} V_{\{a,b\}}$. Then $H_\Lambda = \sum_{ (a,b) \in E_\Lambda} h_e$ is a local Hamiltonian, and clearly $H_\Lambda \ket{\PEPS_{\Lambda, X}} = 0$. $H_\Lambda$ is \textit{frustration-free}: i.e. $h_e\ket{\PEPS_{\Lambda,X}}=0$ for all $e \in E_\Lambda$.

It will be very important for us to talk about sub-regions of the lattice $A\subseteq \Lambda$, and
to consider the associated local ground subspace $\cG_A=\{\ket{\varphi} \in \cH_\Lambda\, |\, H_A\ket{\phi}=0\} = \operatorname{Im} V_A$, for $H_A=\sum_{e\in E_A}h_e$. We will denote with $P_A$ the orthogonal projector on $\cG_A$. Because of frustration freeness, for any $A\subseteq B\subseteq \Lambda$, we have $\cG_\Lambda\subseteq\cG_B\subseteq\cG_A$, and therefore $P_A P_B = P_B = P_B P_A$.

At times, we will need to refer to Hamiltonians both in the bulk (2D) and at the boundary (1D). In order to avoid confusion, we will always denote one dimensional boundary Hamiltonians by the letters $Q,R,S,T$, while the parent Hamiltonian of the PEPS will always be referred to as $H$.

\subsection{Boundary states of PEPS}\label{sec:bnd}

\begin{figure}[h]
\centering
  \includegraphics[scale=0.40]{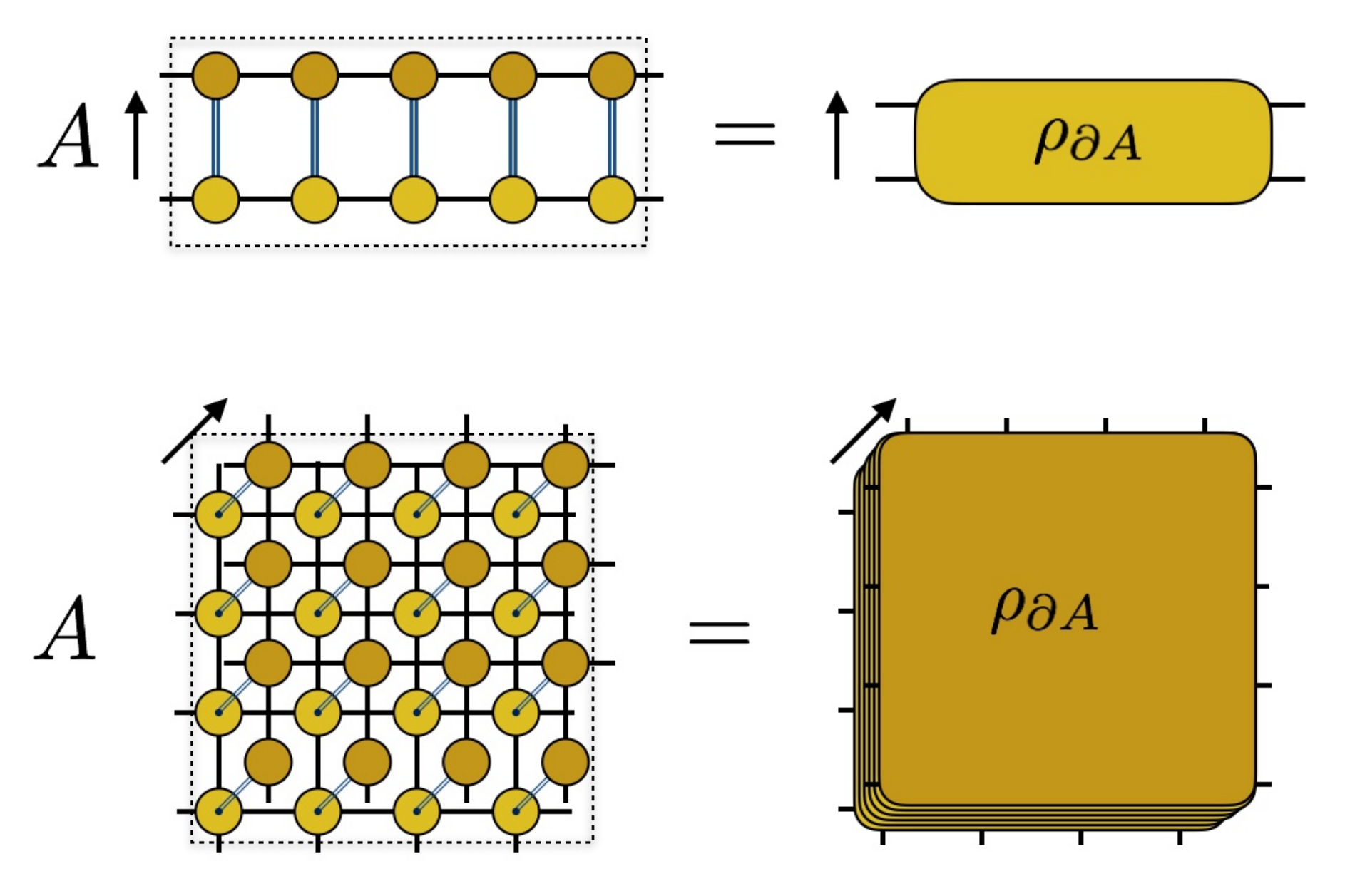}
    \caption{Setup of the boundary state in one and two dimensions. The arrows indicate the input and output directions for the (unnormalized) boundary density matrix $\rho_{\partial A}$. }
    \label{fig:bnd}
\end{figure}

The main conceptual contribution of this paper is that `boundary states' play a very important role in the analysis of ground state projectors for PEPS. These will be (unnormalized) positive operators acting on the virtual space associated with the edges connecting a region $A$ and its complement.
They are obtained by contracting the physical indices inside $A$, and leaving the virtual indices at the boundary open, as depicted in Fig.~\ref{fig:bnd}.
\begin{definition}
For a finite region $A\subseteq\Lambda$, the boundary state of $A$ is
\be\rho_{\partial A}:=V^\dag_AV_A \in \cB(\cH_{\partial A}).\ee
Moreover, we define the following linear operator $W_A : \cH_{\partial A} \to \cH_A$
\be W_A = V_A \rho_{\partial A}^{-1/2},\ee
where the inverse is taken on the support of $\rho_{\partial A}$ if it is not full rank.
\end{definition}

\begin{rem}\label{remark:boundaries}
Some properties of $\rho_{\partial A}$ and $W_A$ follow immediately from the definition
\begin{enumerate}
  \item $\rho_{\partial A}$ is positive semi-definite;
  \item $\ker \rho_{\partial A} = \ker V_A$, and in particular $\rho_{\partial A} >0$ if the PEPS is injective;
  \item $W_A W_A^\dag = V_A \rho_{\partial A}^{-1} V_A^\dag = P_A$;
  \item $W_A^\dag W_A = \1_{(\ker V_A)^{\perp}}$, and therefore $W_A$ is a unitary from $(\ker V_A)^\perp$ to $\operatorname{Im} V_A$, and a partial isometry from $\cH_{\partial A}$ to $\cH_A$ (an isometry if the PEPS is injective).
  \end{enumerate}
\end{rem}

The only point which might not be immediately clear from the definition is the fact that $P_A = W_A W_A^\dag$:
this can be shown by observing that $W_A W_A^\dag$ is a projector, which commutes with $P_A$ since $P_A V_A = V_A$, and has exactly the same image space as $V_A$ (and thus $P_A$).

\begin{rem}
  The entanglement spectrum is the spectrum of the reduced density matrix of a pure state
  \cite{li2008ES}. In the case where $\Lambda$ has no outgoing edges, the entanglement spectrum is
  related to the boundary state in the following way: call $A^{c}=\Lambda\setminus A$ and note that
  \bea \tr_{A^c}[\ket{\PEPS_\Lambda}\bra{\PEPS_\Lambda}]&=&V_A\rho_{\partial A^c}V^\dag_A\nonumber \\
  &=& W_A \rho^{1/2}_{\partial A} \rho_{\partial A^c}\rho^{1/2}_{\partial A} W_A^\dag.\eea $W_A$ is
  an isometry, so the spectrum of the reduced state on $A$ is equal to the spectrum of
  $\rho^{1/2}_{\partial A} \rho_{\partial A^c}\rho^{1/2}_{\partial A}$. Here, $\rho_{\partial A}$
  and $\rho_{\partial A^c}$ are different states that live on the same boundary
  $\partial A = \partial A^c$. $\rho_{\partial A}$ is constructed by contracting all of the physical
  indices in $A$, while $\rho_{\partial A}$ is constructed by contracting the physical indices outside of
  $A$. In certain, very special setups, these two states end up being equal, and so the entanglement spectrum, is equal to the spectrum of
  $\rho_{\partial A}^2$. This is the case in Ref.~\cite{ciracES}, where the system has periodic
  boundary conditions, and it is split exactly in two.
\end{rem}

\subsection{The martingale condition and the spectral gap}
The main result of the paper is to show that under certain conditions, the parent Hamiltonian of a PEPS is gapped. In order to show this, we will invoke an equivalence theorem between the approximate commutativity of ground state projectors and the spectral gap of a frustration free  Hamiltonian \cite{nachtergaele1996,Angelo}. This equivalence has adopted the name `\textit{the martingale method}' although its connection to martingales in probability theory is tenuous at best.
We start by defining the \textit{spectral gap} of a Hamiltonian $H$.

\begin{definition}\label{def:spectral-gap}
  Let $\{H_\Lambda: H_\Lambda \in \cB(\cH_\Lambda) \}_\Lambda$ be a collection of Hamiltonians indexed by some set of regions
  $\Lambda\subset G$. The \emph{spectral gap} of $H_\Lambda$, which we will denote by
  $\lambda(\Lambda)$, is defined as the difference between the two smallest distinct eigenvalues of
  $H_\Lambda$.
  We say that that the family $\{H_\Lambda\}_\Lambda$ is \emph{gapped} if
  \begin{equation}
    \inf_{\Lambda} \lambda(\Lambda) > 0.
  \end{equation}
\end{definition}
It is clear that the condition of being gapped is non trivial only in the case of an infinite family of
Hamiltonians, and in particular we are interested in showing that $\lambda(\Lambda)$ is lower bounded
by a constant independent of the volume $|\Lambda|$ as $\Lambda$ tends to $G$.
We will only consider families of Hamiltonians which are \emph{local} (i.e. which can be constructed
as sum of local interactions) and \emph{frustration-free} (i.e. for which groundstates are the
eigenvalues with minimal energy of all the local interaction terms).
\begin{definition}\label{def:local-Hamiltonian}
  Fix $r\in \bN$. A Hamiltonian $H_\Lambda$ is said to be \emph{$r$-local} if it can be decomposed as
  \begin{equation}
    H_\Lambda = \sum_{Z \subset \Lambda} h_Z^\Lambda,
  \end{equation}
  where each $h_Z^\Lambda \in \cB(\cH_Z)$ is Hermitian, and moreover $h^\Lambda_Z=0$ if $\diam Z >
  r$.
  The value $r$ will be called the \emph{range} and $J = \sup_{Z}\norm{h_Z}$ the \emph{strength} of $H_\Lambda$.
  Moreover, a family $\{H_\Lambda\}_\Lambda$ of Hamiltonians will be said to be \emph{local} if
  there is a choice of $r$ and $J$ such that $H_\Lambda$ is $r$-local with strength less than $J$,
  uniformly in $\Lambda$.
\end{definition}

In a slight abuse of language, we will say that a (single) Hamiltonian $H_\Lambda$ is \emph{local}
(without specifying the range) if it belongs to a local family of Hamiltonians.
We will also restrict to Hamiltonians where interactions terms are given as
projectors, and which are frustration free. Assuming that the interactions are
projections is not restrictive for finite range Hamiltonians, as long as each of
the local terms $h^\Lambda_Z$ can be lower bounded by the projector on its
range, up to a constant independent of $\Lambda$ and $Z$, since this will not
change the low-energy properties of the Hamiltonian.
\begin{definition}[Frustration free]\label{def:ff-Hamiltonian}
  Let $H_\Lambda = \sum_{Z} h^\Lambda_Z$ be a $r$-local Hamiltonian, and let $P_\Lambda$ be the projector on
  its groundstate space. We say that $H_\Lambda$ is frustration-free if $h^\Lambda_Z$ is an orthogonal projector and $h_Z^\Lambda P_\Lambda= 0$ for all $Z$.
\end{definition}
In the case of frustration free Hamiltonians, the spectral gap is simply the smallest non-zero eigenvalue of the Hamiltonian.

We can now state the theorem relating ground state projectors to the spectral gap of $H_\Lambda$.
\begin{figure}[h]
\centering
 \includegraphics[scale=0.35]{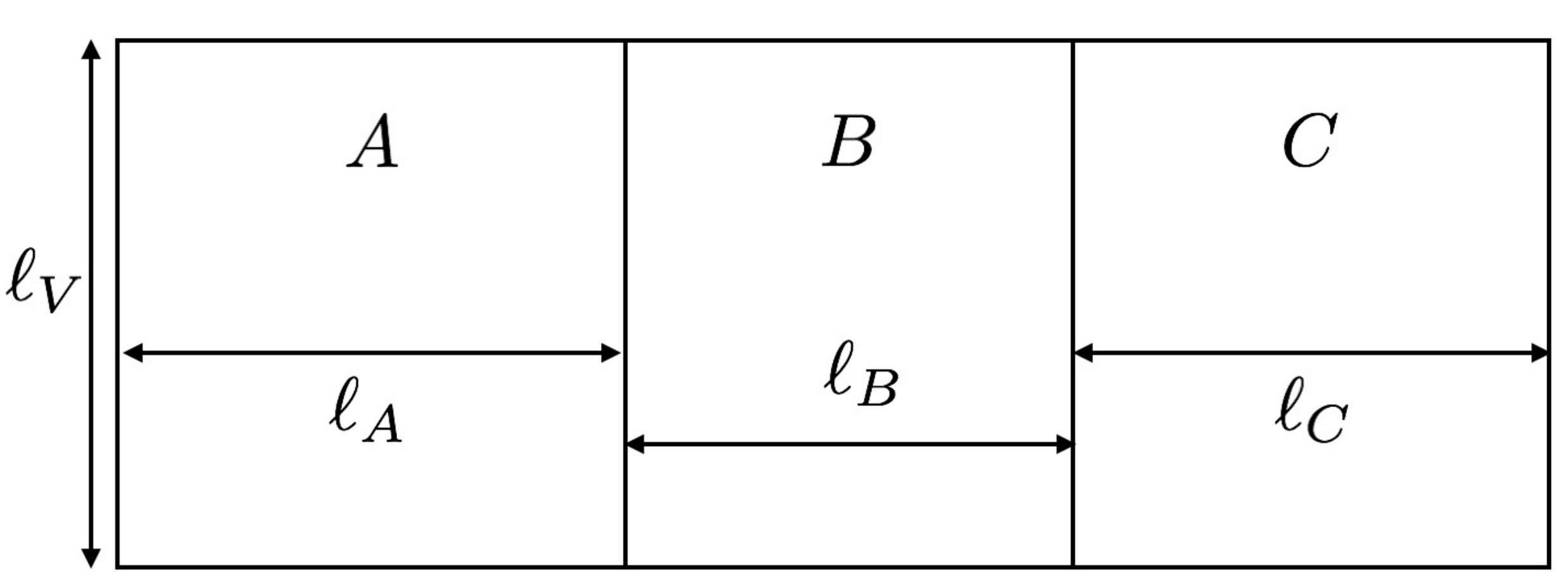}

\caption{Setup for the martingale condition in Theorem \ref{thm:martingale}.}
    \label{fig:ABC}
\end{figure}

\begin{thm}[\cite{Angelo}]\label{thm:martingale}
  Let $\{H_\Lambda\}_\Lambda$ be a family of local, frustration-free Hamiltonians defined on square
  regions $\Lambda\subset G$ (in 2D) whose local terms are orthogonal projections.
  If there exist positive constants $c$ and $\alpha$ and $\beta \in (0,1)$, and if for every three
  adjacent rectangles $A, B, C\subset\Lambda$ as depicted in Fig.~\ref{fig:ABC}, such that the width
  of $B$ separating $A$ and $C$ is $\ell_B\geq (L^*)^\beta $ where
  $L^* = \max\{\ell_V,\ell_A,\ell_C\}$, the following holds
\begin{equation}\label{eqn:martingale}
  \norm{P_{AB}P_{BC}-P_{ABC}} \leq c \ell_B^{-\alpha},
\end{equation}
then $\lambda(\Lambda)$ is bounded from below by a constant independent of $|\Lambda|$. Conversely,
if $\{H_\Lambda\}_\Lambda$ is gapped, then for every $\Lambda=ABC$
\be\norm{P_{AB}P_{BC}-P_{ABC}}\leq c e^{-\ell_B/\xi},\label{eqn:martingale2}\ee
for some constants $c,\xi$.
\end{thm}

See Ref. \cite{Angelo} for a detailed proof and discussion of the theorem in any dimension. Eqn.~(\ref{eqn:martingale}) is referred to as the martingale condition. The norm on the r.h.s. of the
equation can be understood as a strong measure of correlations in the groundstate. One way of seeing this
is to observing that the expression is not small if there is a completely delocalized excitation: a state not in the
groundstate of $ABC$, but which cannot be recognized as such by looking at regions $AB$ and $BC$
alone (so that the excitation is ``hidden'' and can only be measured if we have access to
$A$ and $C$ at the same time).  Theorem \ref{thm:martingale} then states that if this particular measure of
correlations is decaying sufficiently fast in the overlap between regions, then the system is gapped
and the decay is actually exponential. The usefulness of this characterization is due to the fact
that, in the case of PEPS, a complete description of the groundstate subspace is available.

In the remainder of the paper, we will focus on proving that local boundary states imply the
martingale condition, which will immediately imply that the parent Hamiltonian is gapped. An
immediate corollary of Theorem \ref{thm:martingale} is that Eqn.~(\ref{eqn:martingale}) implies
exponential decay of correlation in the ground state $\ket{\Psi_\Lambda}$ of $H_\Lambda$ if the ground state is unique:
\be |\bra{\Psi_\Lambda} f^\dag g\ket{\Psi_\Lambda}-\bra{\Psi_\Lambda} f^\dag \ket{\Psi_\Lambda}\bra{\Psi_\Lambda} g\ket{\Psi_\Lambda}|\leq \sqrt{\bra{\Psi_\Lambda} f^\dag f \ket{\Psi_\Lambda}\bra{\Psi_\Lambda} g^\dag g\ket{\Psi_\Lambda}}\, \gamma^{d(f,g)},\ee
where $d(f,g)$ is the lattice distance between the supports of operators $f,g$ and $\gamma \in (0,1)$ is some constant \cite{aharonov2011}.

\section{The gap theorem for injective PEPS}\label{sec:main}
\subsection{Approximate factorization of the boundary states}

In this section, we will give sufficient conditions on the boundary states such that the martingale condition (Eqn.~(\ref{eqn:martingale})) holds for any regions $ABC$ as in Fig.~\ref{fig:ABC}.

\begin{figure}[h]
\centering
 \includegraphics[scale=0.35]{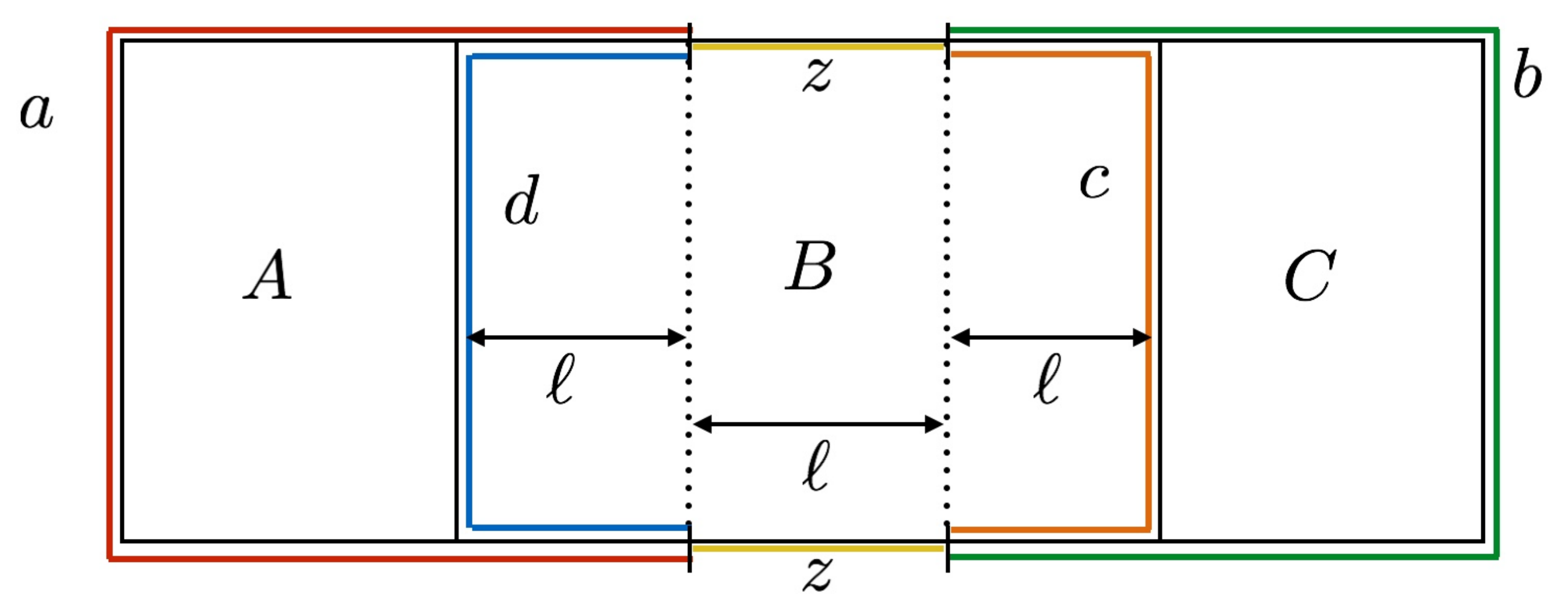}

    \caption{Regions on the boundaries are labeled with lower case roman letters. Both
      $a,d$ and $c,d$ overlap on a region of length $\ell$, which is also a lower bound for the size of
      region $z$. All possible boundaries of
      $A,B,C$ can now be reconstructed in terms of $a,b,c,d,z$, e.g. $\partial AB=a\cup
      z\cup c\equiv azc$. }
    \label{fig:bdregions}
\end{figure}

We consider a specific set of rectangles $ABC$ and assume that the width of $B$ is at least $3\ell$ for some reference length scale $\ell$. In Fig.~\ref{fig:bdregions}, we label regions of the boundaries
with lower case roman letters. It is understood that $a,d$ overlap on a region of length at least
$\ell$, and the same holds true for $c,d$, while they are all disjoint from region $z$. The
boundaries of regions $A,B,C$ can now be reconstructed in terms of $a,b,c,d,z$. For instance,
$\partial AB=a\cup z\cup c\equiv azc$.  Given the notation introduced above, we consider four
invertible matrices $\Delta_{az},\Delta_{zb},\Omega_{dz},\Omega_{zc}$, and define a set of operators
\begin{align}
\sigma_{\partial ABC}&:=\Delta_{zb}\Delta_{az},
& \sigma_{\partial B}&:=\Omega_{zc}\Omega_{dz},\nonumber\\
 \sigma_{\partial AB}&:=\Omega_{zc}\Delta_{az},
&\sigma_{\partial BC}&:=\Delta_{zb}\Omega_{dz}.\label{eqn:sigma}
\end{align}

The operators $\sigma$ are understood as approximate (unnormalized) boundary states that can be
factorized into two overlapping (possibly non-Hermitian) operators, in much the same way as Gibbs states of commuting local Hamiltonians \footnote{Note that in the case of isometric PEPS \cite{MPSrep}, in which the parent Hamiltonian is made of commuting terms, the boundary states are indeed Gibbs states of commuting interactions and have  exactly the form \eqref{eqn:sigma}. As expected, in this case $\epsilon=0$  in Theorem \ref{thm:main}.}. All operators involved are invertible. It is important to note that by construction $\sigma_{\partial BC}\sigma^{-1}_{\partial B}\sigma_{\partial AB}=\sigma_{\partial ABC}$, and $\sigma^{-1}_{\partial AB}\sigma_{\partial B}\sigma^{-1}_{\partial BC}=\sigma^{-1}_{\partial ABC}$.
The support of each operator is indicated by their subscript according to the labels in Fig.~\ref{fig:bdregions}.

\begin{definition}\label{def:qfact}
Let regions $ABC$ be as in Fig.~\ref{fig:ABC}, and let $\rho_{\partial AB},\rho_{\partial B}, \rho_{\partial BC}, \rho_{\partial ABC}$ be the boundary states of regions $AB$, $B$, $BC$ and $ABC$ respectively.
We say that the boundary states are $\epsilon$-\emph{approximately factorizable} with respect to regions $ABC$, if there exist states  $\{\sigma_{\partial AB},\sigma_{\partial B}, \sigma_{\partial BC}, \sigma_{\partial ABC}\}$ with decomposition as in \eqref{eqn:sigma}
such that the following conditions hold:
\begin{align}
  \norm{\rho^{1/2}_{\partial R}\sigma^{-1}_{\partial R}\rho^{1/2}_{\partial R}-\1} &\le
  \epsilon \quad \text{for } R \in \{ ABC, AB, BC\}; \label{QF1} \\
  \norm{\rho^{-1/2}_{\partial B}\sigma_{\partial B}\rho^{-1/2}_{\partial B}-\1} &\le
  \epsilon. \label{QF2}
\end{align}

\end{definition}

We can now state the first main result of the paper:
\begin{thm}\label{thm:main}
If the boundary states of regions $ABC$ are $\epsilon$-approximately factorizable for some $\epsilon
\le 1$, then
$\norm{P_{AB}P_{BC}-P_{ABC}}\leq 8 \epsilon$.
\end{thm}

By Theorem \ref{thm:martingale}, we get that if for any sufficiently large regions $ABC$, where $B$
has diameter $\ell$, the boundary states are $\epsilon(\ell)$-approximately factorizable, for a sufficiently fast decaying  $\epsilon(\ell)$, then the parent Hamiltonian of the PEPS is gapped.

In order to prove Theorem \ref{thm:main}, we first start by showing that the measure of distance
between $\rho_{\partial X}$ and $\sigma_{\partial X}$(for $X$ being $AB$, $BC$, or $ABC$) is the ``correct'' one for the application we need.
\begin{lemma}\label{lemma:main}
For a region $X \subset \Lambda$, let $\rho_{\partial X}$ be the boundary state and $\sigma_{\partial X}$ another invertible operator on $\cH_{\partial X}$. Let $\tilde P_X = V_X \sigma_{\partial X}^{-1} V_X^\dag$. Then
\begin{align}
\norm{\tilde P_X} &=
                 \norm{\rho_{\partial X}^{1/2} \sigma_{\partial X}^{-1} \rho_{\partial X}^{1/2}} \label{eq:lemma-main-1}, \\
\norm{P_X - \tilde P_X} &=
                       \norm{\rho_{\partial X}^{1/2} \sigma_{\partial X}^{-1} \rho_{\partial X}^{1/2} - \1} \label{eq:lemma-main-2}.
\end{align}
\end{lemma}
\begin{proof}
We first recall that $P_X = V_X \rho_{\partial X}^{-1} V_X^\dag = W_X W_X^\dag$.
Therefore, we can rewrite $\tilde P_X$ as follows:
\[ \tilde P_X = V_X \sigma_{\partial X}^{-1} V_X^\dag =
V_X \rho_X^{-1/2} \left(\rho_X^{1/2} \sigma_X^{-1} \rho_X^{1/2} \right) \rho_X^{-1/2} V_X^\dag =
W_X \left(\rho_X^{1/2} \sigma_X^{-1} \rho_X^{1/2} \right) W_X^\dag.
 \]
Then Eq. \eqref{eq:lemma-main-1} follow immediately from the fact that $W_X$ is an isometry, and similarly does Eq. \eqref{eq:lemma-main-2}:
\[ \norm{P_X - \tilde P_X} = \norm{W_X \left(\1 - \rho_X^{1/2} \sigma_X^{-1} \rho_X^{1/2} \right) W_X^\dag }
= \norm{ \1 - \rho_X^{1/2} \sigma_X^{-1} \rho_X^{1/2}}.
 \]
\end{proof}

\begin{figure}[h]
  \centering
  \includegraphics[scale=0.5]{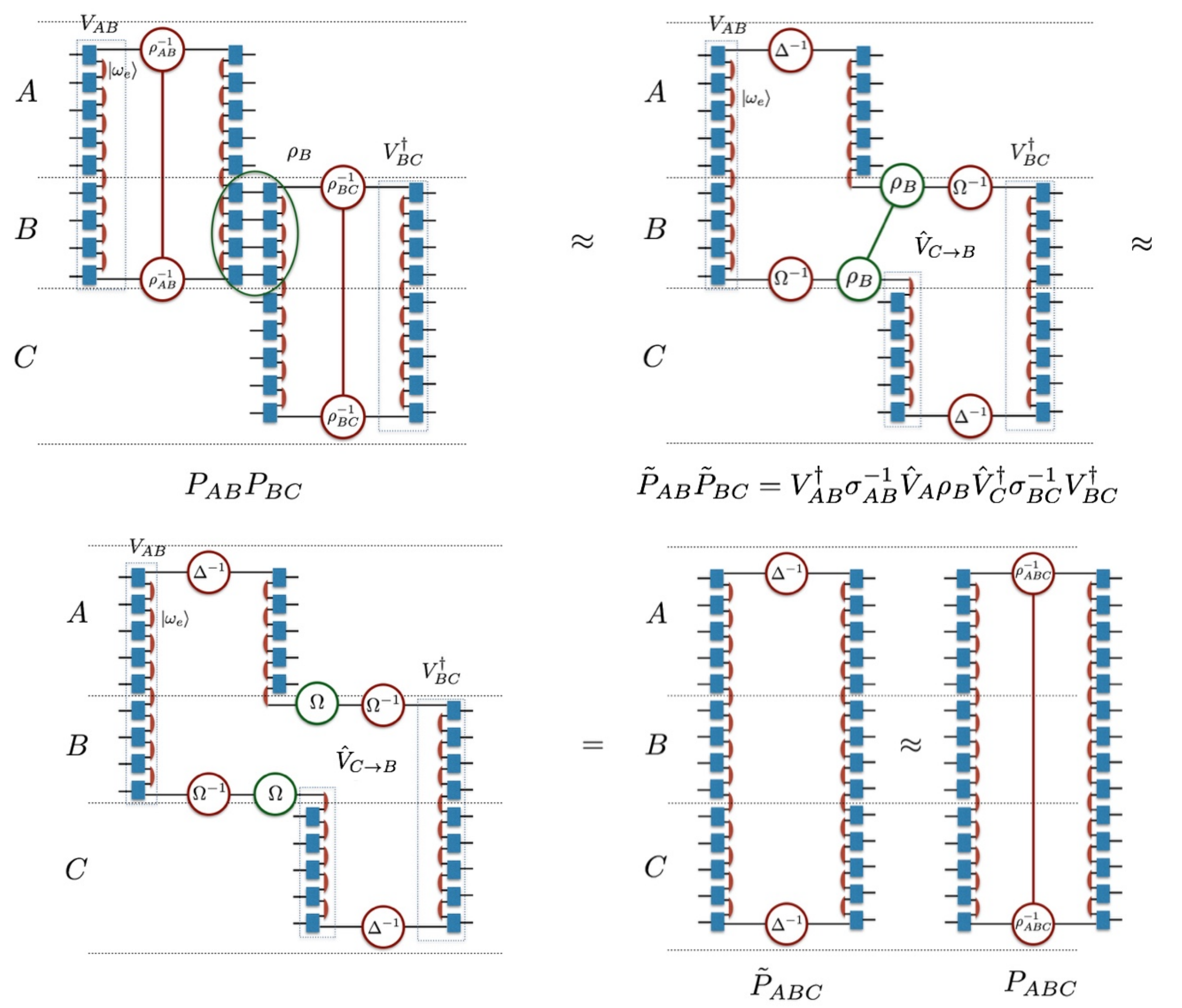}
    \caption{Sketch of the proof in 1D. In the first step we simply write out the product of $P_{AB}P_{BC}$ in terms of the operators $V$ and the boundary states $\rho$. Then, using Eqn. \eqref{QF2}, we approximate the boundary states of $AB$ and $BC$ by factorized operators $(\Omega,\Delta)$, giving us $\tilde{P}_{AB}\tilde{P}_{BC}$. In the next step we approximate the boundary state of $B$ by the factorized from Eqn.~(\ref{QF1}). Then the two ends of the respective boundaries cancel out, and we are left with $\tilde{P}_{ABC}$. In the final step we approximate the factorized boundary state on $ABC$ with the real boundary state using Eqn.~(\ref{QF1}) again.}
    \label{fig:sketch}
\end{figure}

We are now ready to prove Theorem \ref{thm:main}.
\begin{proof}[Proof of Theorem \ref{thm:main}]
A sketch of the proof is provided in Fig.~\ref{fig:sketch} for one dimensional systems, but provides the necessary intuition to follow the steps below.
We  start by setting some notation: given the regions $A$ and $B$, we denote by $\hat V_{A\to B}$ a
modified version of $V_A$, where we have transposed from inputs to outputs the sites corresponding to
the boundary shared between $A$ and $B$. $\hat V_{A\to B}$ then is a linear map from $\cH_{\partial
  A\setminus \partial B}$ to $\cH_{A} \otimes \cH_{\partial A \cap \partial B}$, and we can write
$V_{ABC} = V_{BC} \hat V_{A\to B}$ and $V_{AB} = V_B \hat V_{A\to B}$. We will do the same to define $\hat V_{C\to B}$.
Using such notation, we see that
\[ \tilde P_{AB} \tilde P_{BC} =
V_{AB} \sigma_{\partial AB}^{-1} \hat V_{A\to B}^\dag V_B^\dag V_B \hat V_{C\to B} \sigma^{-1}_{\partial BC} V_{BC}^\dag =
V_{AB} \sigma_{\partial AB}^{-1} \hat V_{A\to B}^\dag \rho_{\partial B} \hat V_{C\to B} \sigma^{-1}_{\partial BC} V_{BC}^\dag,\]

 while using  \eqref{eqn:sigma} we have that:
\begin{eqnarray*}
\tilde P_{ABC} &=& V_{ABC} \sigma_{\partial ABC}^{-1} V_{ABC}^\dag \\
&=& V_{AB} \hat V_{C\to B} \Delta_{az}^{-1}  \Delta_{zb}^{-1} \hat V_{A\to B}^\dag V_{BC}^\dag \\
&=&V_{AB}  \Delta_{az}^{-1} \hat V_{A\to B}^\dag    \hat V_{C\to B}^\dag  \Delta_{zb}^{-1} V_{BC}^\dag  \\
&=&V_{AB}  \Delta_{az}^{-1} \Omega_{zc}^{-1} \hat V_{A\to B}^\dag  \Omega_{zc} \Omega_{dz}  \hat V_{C\to B}^\dag  \Omega_{dz}^{-1} \Delta_{zb}^{-1} V_{BC}^\dag \\
&=& V_{AB} \sigma_{\partial AB}^{-1} \hat V_{A\to B}^\dag \sigma_{\partial B} \hat V_{C\to B} \sigma^{-1}_{\partial BC} V_{BC}^\dag.
\end{eqnarray*}
We can therefore bound their difference as follows
\begin{align*}
 \norm{\tilde P_{AB} \tilde P_{BC} - \tilde P_{ABC}} =
 \norm{ V_{AB} \sigma_{\partial AB}^{-1} \hat V_{A\to B}^\dag \left(\rho_{\partial B}-\sigma_{\partial B}\right) \hat V_{C\to B} \sigma^{-1}_{\partial BC} V_{BC}^\dag}  &\\
 =
 \norm{ V_{AB}  \sigma_{\partial AB}^{-1} \hat V_{A\to B}^\dag \rho_{\partial B}^{1/2}\left(\1- \rho_{\partial B}^{-1/2}\sigma_{\partial B} \rho_{\partial B}^{-1/2}\right)  \rho_{\partial B}^{1/2}\hat V_{C\to B} \sigma^{-1}_{\partial BC} } &\\
\le \norm{ V_{AB}  \sigma_{\partial AB}^{-1} \hat V_{A\to B}^\dag \rho_{\partial B}^{1/2} } \,
\norm{ V_{BC} \sigma_{\partial BC}^{-1} \hat V_{C\to B}^\dag \rho_{\partial B}^{1/2}} \,
\norm{ \1- \rho_{\partial B}^{-1/2}\sigma_{\partial B} \rho_{\partial B}^{-1/2}}&.
\end{align*}
In order to bound the first two terms in the r.h.s. of the last equation, we observe that
\begin{multline*}
\left(V_{AB}  \sigma_{\partial AB}^{-1} \hat V_{A\to B}^\dag \rho_{\partial B}^{1/2}\right)
\left(V_{AB}  \sigma_{\partial AB}^{-1} \hat V_{A\to B}^\dag \rho_{\partial B}^{1/2}\right)^\dag
\\ =
V_{AB} \sigma_{\partial AB}^{-1} \hat V_{A\to B}^\dag \rho_{\partial B} \hat V_{A\to B} \sigma_{\partial AB}^{-1} V_{AB}^\dag = \tilde P_{AB}^2,
\end{multline*}
which implies that
\[ \norm{ V_{AB}  \sigma_{\partial AB}^{-1} \hat V_{A\to B}^\dag \rho_{\partial B}^{1/2} } = \norm{ \tilde P_{AB}^2 }^{1/2} \le \norm{ \tilde P_{AB} } .\]
Since the same holds for the terms on $BC$, we have proven that
\begin{equation} \norm{\tilde P_{AB} \tilde P_{BC} - \tilde P_{ABC} } \le \norm{ \tilde P_{AB} } \, \norm{\tilde P_{BC} } \, \norm{ \1- \rho_{\partial B}^{-1/2}\sigma_{\partial B} \rho_{\partial B}^{-1/2} }.
\end{equation}

We can now apply Lemma \ref{lemma:main} and Eqs.~\eqref{QF1} and \eqref{QF2}, to get
\[ \norm{\tilde P_{AB} \tilde P_{BC} - \tilde P_{ABC} } \le (1+\epsilon)^2 \epsilon ,\]
since $\norm{ \tilde P_{AB} } \le \norm{ \tilde P_{AB} - P_{AB} } + \norm{ P_{AB} } \le \epsilon + 1$.

To conclude the proof, we just have to resort one last time to Lemma \ref{lemma:main}, as we have that:
\bea
\norm{P_{AB}P_{BC}-P_{ABC}}\leq&&\norm{\tilde{P}_{AB}\tilde{P}_{BC}-\tilde{P}_{ABC}}+\norm{\tilde{P}_{AB}}~\norm{P_{BC}-\tilde{P}_{BC}}\nonumber\\
&+&\norm{P_{AB}-\tilde{P}_{AB}}~\norm{P_{BC}}+\norm{\tilde{P}_{ABC}-P_{ABC}} \\
\le && (1+\epsilon)^2 \epsilon + (1+\epsilon)\epsilon + 2 \epsilon \le 8 \epsilon.
\nonumber\eea
\end{proof}

\subsection{Matrix Product States}
As is often the case, in one dimension the situation becomes particularly simple since the boundary
is zero dimensional and has two spatially separated ends. We will show that the boundary state of an
injective
MPS is approximately factorizable when the region is long enough. Consider a translationally invariant MPS on a
chain of length $N$:
\begin{equation*}
  \ket{{\rm MPS}}= \sum_{j_1,\cdots,j_N = 1}^d\tr{T_{j_1}\cdots T_{j_N}}\ket{j_1,\cdots,j_N},
\end{equation*}
where $\{T_{j}\}_{j=1}^d$ is a collection of $d$ matrices of size $D\times D$.

Many of the properties of an MPS can be succinctly described by the transfer operator
$\bE(f)=\sum_{j} T_{j} f T_{j}^\dag$, which maps virtual bonds to virtual bonds by contracting one
single physical bond (see Fig.~\ref{fig:MPS}). The boundary state $\rho_{\partial A}$ of a region $A$ of length $m$ is the Choi matrix of the
$m$-th power of $\bE$:
\begin{equation}
  \rho_{\partial A} = \bE^m \otimes \id ( \proj{\Omega}), \quad \ket\Omega = \sum_{i=1}^D \ket{i,i}.
\end{equation}
Note that by construction $\bE$ is a completely positive map, so its Choi matrix is a positive operator.

\begin{figure}[ht]
\centering
  \includegraphics[scale=0.4]{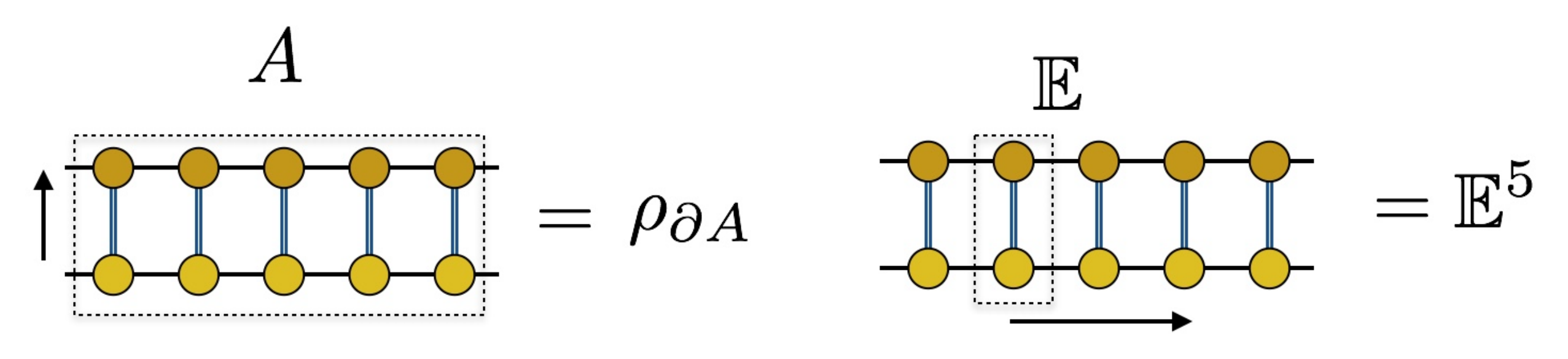}
    \caption{In one dimension, the boundary state can be reinterpreted as a power of the transfer operator when we flip the indices from reading up-down to left-right.}
    \label{fig:MPS}
\end{figure}

If we replace the matrices $T_{j}$ with matrices $Z T_{j} Z^{-1}$, for some invertible matrix
$Z$, it is easy to see that the state $\ket{\rm MPS}$ is left invariant \cite{MPSrep} (although the
boundary state $\rho_{\partial A}$ will not: it will be mapped to
$(Z\otimes Z^{-1})^\dag \rho_{\partial A}(Z \otimes Z^{-1})$). This operation is usually referred to
as ``choosing the gauge'' of the MPS representation. It is a well known fact that, in the case of
injective MPS, there is a choice of gauge that makes the transfer operator a trace-preserving map
\cite{MPSrep}. This allows us to prove the following property of the boundary states of an injective
MPS.
\begin{lemma}\label{lemma:mps-state-convergence}
  Given an injective MPS, there is a representation with the following property: there exist a full
  rank state $\sigma \in \cB(\bC^d)$, a positive constant $c$ and a $\gamma \in (0,1)$ such that,
  for every chain $A$ of length $m$ it holds that
  \begin{equation}\label{eq:mps-state-convergence}
    \norm{\rho_{\partial A} - \sigma \otimes \1}_1 \le c \gamma^{m},
  \end{equation}
  where $\rho_{\partial A}$ is the boundary state on $A$.
\end{lemma}
\begin{proof}
  The MPS is injective on a
  chain of length $m$ if and only if the set
  $\{T_{j_1} \cdots T_{j_m}\, |\, j_1, \dots, j_m = 1 \dots d\}$ spans the full matrix algebra
  $\cB(\bC^{D})$: but since that is exactly the set of Kraus operator of $\bE^m$,
  it implies that $\bE$ is primitive (a primitive linear map is an
  irreducible positive map with trivial peripheral spectrum) \cite{WolfNotes}. The adjoint map
  $\bE^*(g) = \sum_j T_{j}^\dag g T_{j}$ will share the same property, so that we can find a
  positive-definite operator $y$ such that $\bE^*(y) = \lambda y$ where $\lambda>0$ is the largest
  eigenvalue of $\bE$, which is also simple. We can then change the matrices in the MPS
  representation, replacing each $T_{j}$ with $\frac{1}{\lambda} y^{1/2} T_{j} y^{-1/2}$, without
  changing the state $\ket{\rm MPS}$ (apart from its normalization). With this choice of gauge,
  the transfer operator is trace preserving, since $\bE^*(\1) = \1$. Moreover, the sequence of linear
  maps $\bE^m$ will converge, as $m$ goes to infinity, to another completely positive, trace
  preserving map $\bE^\infty$, which is of the form $\bE^\infty(f) = \tr{f} \sigma$, for some fixed
  $\sigma >0$. The Choi matrix of $\bE^\infty$ is given by

  \begin{eqnarray}
    \tau_{\bE^\infty} &=& \lim_{m\rightarrow\infty}\bE^m \otimes \id (\ket{\Omega}\bra{\Omega})\nonumber \\
                      &=& \sum_{jk} \lim_{m\rightarrow\infty}\bE^m(\ket{j}\bra{k}) \otimes \ket{j}\bra{k}\nonumber\\
                      &=& \sum_j \sigma\otimes \ket{j}\bra{j}=\sigma\otimes\1\nonumber.
  \end{eqnarray}

  Since the dimension $D$ is fixed, by the Jordan decomposition there exists some $\gamma \in (0,1)$ and some constant $c>0$ such that
  \[ \norm{\bE^m - \bE^\infty}_{1\to 1, cb} \le c\gamma^m, \]
  where $\norm{\cdot}_{1\to 1,cb}$ is the completely bounded $1\to 1$ norm, also known as the diamond norm.
  Furthermore, the difference between the Choi matrix of two channels is bounded by their difference in diamond norm $\norm{\tau_T-\tau_S}_1\leq \norm{T-S}_{1\rightarrow 1,cb}$.
  This is obvious since $\norm{T-S}_{1\to 1,cb}=\sup_{\rho}\norm{T\otimes\id(\rho)-S\otimes\id(\rho)}_1$.
  By this argument, we get that for this specific choice of gauge, if $A$ is a chain of length $m$, it
  holds that
  \[  \norm{\rho_{\partial A}-\sigma\otimes\1}_1\leq c \gamma^m.\]
\end{proof}

Since we are working in finite dimension (every $\rho_{\partial A}$ lives in the same space
$\cB(\bC^D)$ independently of $m$), the fact that $\rho_{\partial A}$ converges in the trace norm to
a product state is sufficient to prove that is approximately factorizable in the sense of equations \eqref{QF1}
and \eqref{QF2}. To show this, we first need the following lemma.

\begin{lemma}\label{lemma:comparing_norms}
  If $X, Y$ are positive operators, and $Y$ is invertible,
  it holds that:
  \begin{equation}\label{eq:comparing_norms}
    \norm{X^{1/2} Y^{-1} X^{1/2} - \1 } \le Y_{\min}^{-1} \norm{X -Y},
  \end{equation}
  where $Y_{\min}$ is the smallest eigenvalue of $Y$.
\end{lemma}
\begin{proof}
  First of all, we rewrite $X^{1/2} Y^{-1} X^{1/2} - \1$ as $X^{1/2}(Y^{-1} - X^{-1})X^{1/2}$, where
  the inverse of $X$ is taken on its support. $X^{1/2}(Y^{-1} - X^{-1})X^{1/2}$ is a normal operator with the same spectrum as
  $(Y^{-1} - X^{-1})X$. Since the spectral radius of an operator is a lower bound to its operator
  norm, and for normal operators equality holds, we have that:
  \[ \norm{X^{1/2}(Y^{-1} - X^{-1})X^{1/2}} \le \norm{(Y^{-1} - X^{-1})X}.\]
  (See also \cite[Proposition IX.1.1]{Bhatia1997}).
  We can now rewrite $(Y^{-1} - X^{-1})X$ as $Y^{-1}(X-Y)$, so that we obtain
  \[ \norm{X^{1/2} Y^{-1} X^{1/2} - \1 } \le \norm{Y^{-1}(X-Y)} \le Y_{\min}^{-1} \norm{X -Y}. \]
\end{proof}

To conclude that $\rho_{\partial A}$ is approximately factorizable, we can then apply Lemma
\ref{lemma:comparing_norms} to $\rho_{\partial A}$ and $\sigma \otimes \1$, and we obtain that
\[
  \norm{\rho_{\partial A}^{1/2}(\sigma^{-1} \otimes \1) \rho_{\partial A}^{1/2} - \1 }
  \le \sigma_{\min}^{-1} \norm{\rho_{\partial A} - \sigma \otimes \1}_1 \le c \sigma_{\min}^{-1} \gamma^{m},
\]
since we can upper bound the operator norm with the trace norm. Therefore, it will decay
exponentially in the length of $A$.
We can more easily bound $\norm{\rho_{\partial A}^{-1/2}(\sigma^{1/2} \otimes \1) \rho_{\partial A}^{-1/2} -\1
}$ by $r_{\partial A}^{-1} \norm{\rho_{\partial A} - \sigma}_1$, where $r_{\partial A}$ is the
minimal eigenvalue of $\rho_{\partial A}$. We can then observe that $r_{\partial A}$
is lower bounded by $\sigma_{\min} - \norm{\rho_{\partial A} - \sigma}$. Therefore, if $m$ is
sufficiently large, we can assume that $r_{\partial A}$ is larger than $\sigma_{\min}/2$, so that
we also have
\[ \norm{\rho_{\partial A}^{-1/2}(\sigma^{1/2} \otimes \1) \rho_{\partial A}^{-1/2} -\1}
  \le c^\prime \sigma_{\min}^{-1} \gamma^{m},\]
for some positive constant $c^\prime$.

\section{The gap theorem for non-injective PEPS}\label{sec:topo}
All of the results so far have been obtained for injective PEPS. Injectivity is essential in guaranteeing that the boundary state $\rho_{\partial X}$ is full rank and hence invertible. In the MPS setting, we also saw that injectivity is sufficient to show approximate factorization.
However, injective PEPS exclude any description of topologically ordered phases. In this section we
consider extensions of injectivity (G-injectivity and MPO-injectivity) that allow for the
description of most known topological phases of gapped spin systems. We show that with a slight
modification of definitions, we can extend Theorem \ref{thm:main} to the setting of
MPO-injectivity.

\subsection{MPO-injectivity}

Injective PEPS can be seen as perturbations of trivial short-range entangled states (products of nearest-neighbor maximally entangled states) via Eqn.~(\ref{eqn:PEPS}). In the same way,  more complicated states, such as the toric code or other topologically-ordered states, can be taken as base states to be perturbed, giving rise to different classes of PEPS.

A first construction is to take as base state  $|{\rm Base}\rangle$ the so called $G$-isometric PEPS \cite{schuch2010}, for a given finite group $G$. They correspond to Kitaev's quantum double models $D(G)$. In particular for the group $G=\mathbb{Z}_2$, $|{\rm Base}\rangle$ is just the toric code. $G$-isometric PEPS  are defined by fixing the bond and physical dimensions respectively to $D=|G|$, $d=D^{\otimes 4}$ and by choosing as PEPS tensor $T_v=\frac{1}{|G|}\sum_{g\in G}L_g^{\otimes 4}$ in Eqn.~(\ref{eqn:PEPS}), with  $L_g$ being the left regular representation of $G$.

It was shown in \cite{cirac2017} that the parent Hamiltonian of a $G$-isometric PEPS is commuting, and that the boundary state $\rho_{\partial A}$ of the PEPS in a region $A\subset\Lambda$, as defined previously in the text (Figure 2), is exactly the projector $J_{\partial A}:=\frac{1}{|G|}\sum_{g\in G}L_g^{\otimes |\partial A|}$. Note that $J_{\partial A}$ can be written as a translational invariant Matrix Product Operator (MPO) with bond dimension $|G|$:

$$J_{N_\text{ sites}}=\frac{1}{|G|}\sum_{g_1, \ldots g_N\in G} {\rm tr}(B_{g_1}\cdots B_{g_N})L_{g_1}\otimes \cdots \otimes L_{g_N}.$$
For that, it is enough to take $B_g=|g\rangle\langle g|$.

Perturbed PEPS of the form $\bigotimes_{v\in \Lambda}Y_v\ket{{\rm Base}_\Lambda}$, with $Y_v$ invertible, where $|{\rm Base}_\Lambda\rangle$ is a $G$-isometric PEPS are called $G$-injective.
The construction of $G$-injective PEPS was generalized first by Buerschaper in \cite{buerschaper2014} and later by Sahinoglu et al in \cite{MPOinjective} considering as initial base state all Levin-Wen string-net states \cite{levin2005}, which are believed to cover all possible 2D non-chiral topological phases. The starting point of the construction by Sahinoglu et al in \cite{MPOinjective} (see \cite{bultinck2017}) is a translational invariant MPO $J_N$ which is a projector for all system size $N$. As shown in \cite{MPOinjective, bultinck2017}, by invoking the fundamental theorem of Matrix Product Vectors \cite{cirac2017}, this induces an algebra of MPO which in turn gives rise to a fusion category. The  state $|{\rm Base}_\Lambda\rangle$ is defined in the same way as for the $G$-injective case: $T_v$ is given by the MPO-projector acting on four sites $J_4$. The resulting state  $\bigotimes_{v\in \Lambda}Y_v\ket{{\rm Base}_\Lambda}$ is called an \emph{MPO-injective PEPS}.

\subsection{Approximate factorization for MPO-injective PEPS}
Some simple properties of MPO-injective PEPS will be sufficient to extend the results of Theorem \ref{thm:main}.
\begin{definition}
  Given a non-injective PEPS, for each region $A\subset \Lambda$ denote by $J_{\partial A} \in \cB(\cH_{\partial A})$ the projector on the complement of the kernel of $\rho_{\partial A}$.
\end{definition}
Note that the kernel of $\rho_{\partial A}$ coincides with the kernel of $V_A$ (see Remark \ref{remark:boundaries}), so that we trivially have  $V_A = V_A J_{\partial A}$.
We need a compatibility condition between the $J_{\partial A}$ acting on overlapping regions. We will use the modified map $\hat V_{A\to B}$ defined in the proof of Theorem \ref{thm:main}.
\begin{figure}[h]
  \centering
  \includegraphics[scale=0.40]{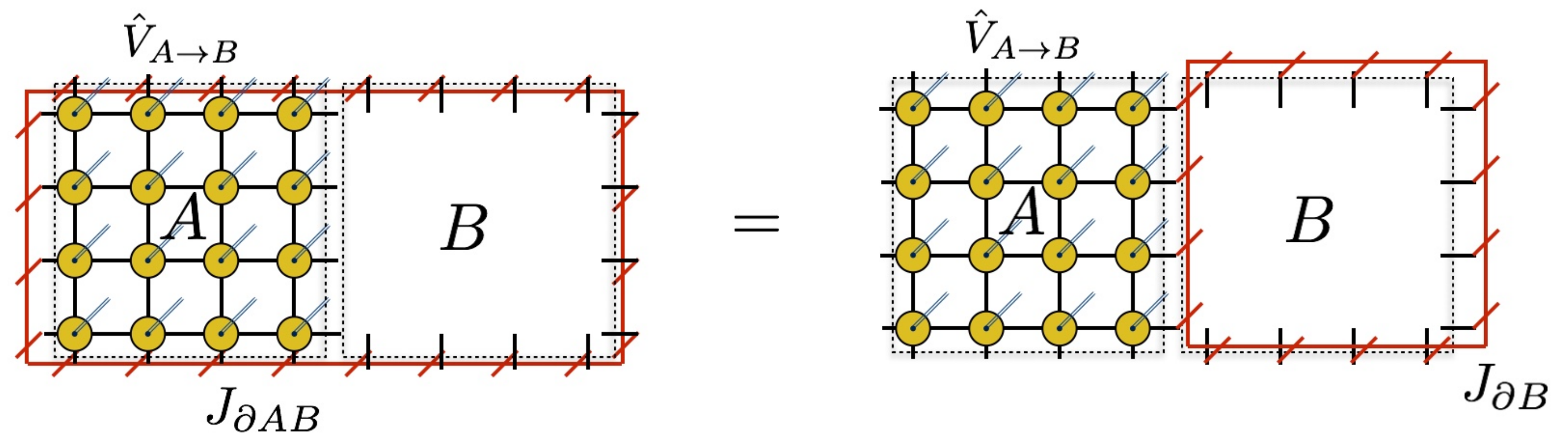}
  \caption{Setup for pulling through condition. The MPO projector $J_{\partial B}$ is indicated in red, and acts on the boundary of $AB$ from the right (bottom in figure), and acts on the boundary of $B$ from the left (top in the figure). }
  \label{fig:topo}
\end{figure}
\begin{definition}[Pulling-through condition]
  A PEPS satisfies the \emph{pulling-through} condition \cite{MPOinjective} if for every pair of contiguous regions $A$ and $B$, it holds that
  \begin{equation}
    J_{\partial B} \hat V_{A\to B}= \hat V_{A\to B} J_{\partial AB}
  \end{equation}
\end{definition}

As in Theorem \ref{thm:main}, it is understood that the virtual indices of $\hat{V}_{A\rightarrow B}$ that touch $B$ are output variables rather than input variables as in $V_A$ (see Fig.~\ref{fig:topo}).

We will now show that the proof of Theorem \ref{thm:main} can be adapted to the case of MPO-injective PEPS if we restrict all the operators living on the boundary to have support on the range of the projectors $J_{\partial A}$.

\begin{definition}\label{def:approx-fact-mpoinj}
  Let $ABC$ be regions as in Fig.~\ref{fig:ABC} and $\rho_{\partial AB}$, $\rho_{\partial BC}$, $\rho_{\partial B}$ and $\rho_{\partial ABC}$ the boundary states of regions $AB$, $BC$, $B$ and $ABC$ respectively. We will say that the boundary states are \emph{$\epsilon$-approximately factorizable} with respect to regions $ABC$, if there exist invertible operators $\Delta_{zb}$, $\Delta_{az}$, $\Omega_{zc}$, $\Omega_{dz}$ with support given by Fig.~\ref{fig:bdregions}, defining operators $\sigma_{\partial ABC}$, $\sigma_{\partial B}$, $\sigma_{\partial AB}$ and $\sigma_{\partial BC}$ as in \eqref{eqn:sigma},
  and the  following holds:
  \begin{align}
    [J_{\partial R}, \sigma_{\partial R}] &= 0; \\
    \norm{\rho^{1/2}_{\partial R}\sigma^{-1}_{\partial R}\rho^{1/2}_{\partial R}-J_{\partial R}} &\le
                                                                                                   \epsilon \quad \text{for } R \in \{ ABC, AB, BC\}; \label{QF1-mpoinj} \\
    \norm{\rho^{-1/2}_{\partial B}J_{\partial B}\sigma_{\partial B}J_{\partial B}\rho^{-1/2}_{\partial B}-J_{\partial B}} &\le
                                                                                                                            \epsilon. \label{QF2-mpoinj}
  \end{align}
\end{definition}

The reason for the change in Eqn. \eqref{QF1-mpoinj}  compared to Eqns. \eqref{QF1} and \eqref{QF2} is clear given the following extension of Lemma \ref{lemma:main}:
\begin{lemma}\label{lemma:main-mpoinj}
  For a region $X\subset \Lambda$, let $\rho_{\partial X}$ be the boundary state and $\sigma_{\partial X}$ another operator invertible on $J_{\partial X}\cH_{\partial X} J_{\partial X}$. Let $\tilde P_X = V_X \sigma^{-1}_{\partial X} V_X^\dag$. Then
  \begin{align}
    \norm{ \tilde P_X } &= \norm{\rho^{1/2}_{\partial X} \sigma^{-1}_{\partial X} \rho^{1/2}_{\partial X}}, \\
    \norm{ \tilde P_X - P_X } &= \norm{\rho^{1/2}_{\partial X} \sigma^{-1}_{\partial X} \rho^{1/2}_{\partial X}- J_{\partial X}}.
  \end{align}
\end{lemma}
Since the proof is identical to that of Lemma \ref{lemma:main}, with the only difference that now $W_X$ is a partial isometry and $W_X W_X^\dag = J_{\partial X}$, so we will omit it. By using this last lemma instead of Lemma \ref{lemma:main} in the proof of Theorem \ref{thm:main}, the proof carries through almost identically: the pulling-through condition guarantees that the projections $J_{\partial X}$ can be moved through the equation as needed. In particular, we have that
\bea
\tilde P_{ABC} &=& V_{ABC}   \sigma_{\partial ABC}^{-1}  V_{ABC}^\dag\nonumber\\
&=& V_{AB} \hat V_{C\to B} J_{\partial ABC} \Delta_{az}^{-1} \Delta_{zb}^{-1} J_{\partial ABC}\hat V_{A\to B}^\dag V_{BC}^\dag\nonumber\\
&=& V_{AB} J_{\partial AB} \Delta_{az}^{-1} \hat V_{A\to B}^\dag \hat V_{C\to B} \Delta_{zb}^{-1} J_{\partial BC} V_{BC}^\dag \nonumber \\
&=& V_{AB} J_{\partial AB} \Delta_{az}^{-1} \Omega_{zc}^{-1} \hat V_{A\to B}^\dag \Omega_{zc} \Omega_{dz} \hat V_{C\to B} \Omega_{dz}^{-1} \Delta{zb}^{-1} J_{\partial BC} V_{BC}^\dag \nonumber\\
&=& V_{AB} J_{\partial AB} \sigma_{\partial AB}^{-1} \hat V_{A\to B}^\dag \sigma_{\partial B} \hat V_{C\to B} \sigma_{\partial BC}^{-1} J_{\partial BC} V_{BC}^\dag  \nonumber\\
&=& V_{AB} \sigma_{\partial AB}^{-1} \hat V_{A\to B}^\dag J_{\partial B} \sigma_{\partial B}J_{\partial B}  \hat V_{C\to B} \sigma_{\partial BC}^{-1}  V_{BC}^\dag.
\eea
The rest of the proof is identical to the non-topological case.

\section{Approximate factorization of 1D thermal states}\label{sec:qfact}
In this section, we consider a class of physically motivated states for which we can show that they
are approximately factorizable, by explicitly constructing operators $\{\Delta,\Omega\}$ satisfying
Eqns.~(\ref{QF1}, \ref{QF2}).

\subsection{Local and quasi-local Hamiltonians}
We define the following properties for
families of operators $f_{\partial A}$ defined for every rectangular subset $A$ of the lattice, and
acting on the boundary Hilbert space $\cH_{\partial A}$.
\begin{definition}
  Let $\{f_{\partial A}\}_{A}$ be a family of operators such that $f_{\partial A}
  \in \cB(\cH_{\partial A})$, where the index $A$ runs over all rectangles $A\subset G$. For each
  $A$, we decompose $f_{\partial A}$ as follows:
  \[ f_{\partial A} = \sum_{Z \subset \partial A} f^{\partial A}_Z, \]
  where each $f^{\partial A}_Z$ is supported on $Z$ (such decomposition is always trivially possible). Moreover,
  if every $f_{\partial A}$ is Hermitian, we require every $f^{\partial A}_Z$ to be Hermitian as well. We will then say that:
  \begin{description}
  \item[locality:\label{def:locality}] the family $\{f_{\partial A}\}_{A}$ is \emph{local}, if there exist an integer $k^*$
    and a constant $J>0$ such that
    \[ \sup_{\partial A} \sup_Z \norm{f^{\partial A}_Z} \le J, \]
    and moreover $f_Z = 0 $ if the diameter of $Z$ is larger than $k^*$.
    The value $k^*$ will denote the \emph{range} of $f_{\partial A}$, while $J$ will be the
    \emph{strength} of $f_{\partial A}$. We will also say that $f_{\partial A}$ is \emph{$k^*$-local}.
  \item[quasi-locality:\label{def:quasi-locality}] the family $\{f_{\partial A}\}_{A}$ is \emph{quasi-local}, if
    \begin{equation}
      \forall x \ge 1, \quad \sup_{A} \sup_{u\in \partial A}\sum_{Z\ni u} x^{\diam{Z}} \norm{f^{\partial A}_Z} < \infty.
    \end{equation}
  \end{description}
\end{definition}

The definition of quasi-locality implies that, for each $\partial A$, the norm
$\norm{f_Z^{\partial_A}}$ decays in the diameter of $Z$ faster than any exponential. This is quite
a stronger requirement than what is usually made. In the setting of local Hamiltonians, one usually considers the case of exponential or faster than polynomial decaying interactions,
while the quasi-local algebra for quantum spin systems in the thermodynamic limit usually contains
any norm-convergent sequence of local operators \cite{Bratteli_1987}. The motivation for our choice of a stronger notion  will be clear in the next section.

\begin{figure}[h]
  \centering
  \includegraphics[scale=0.35]{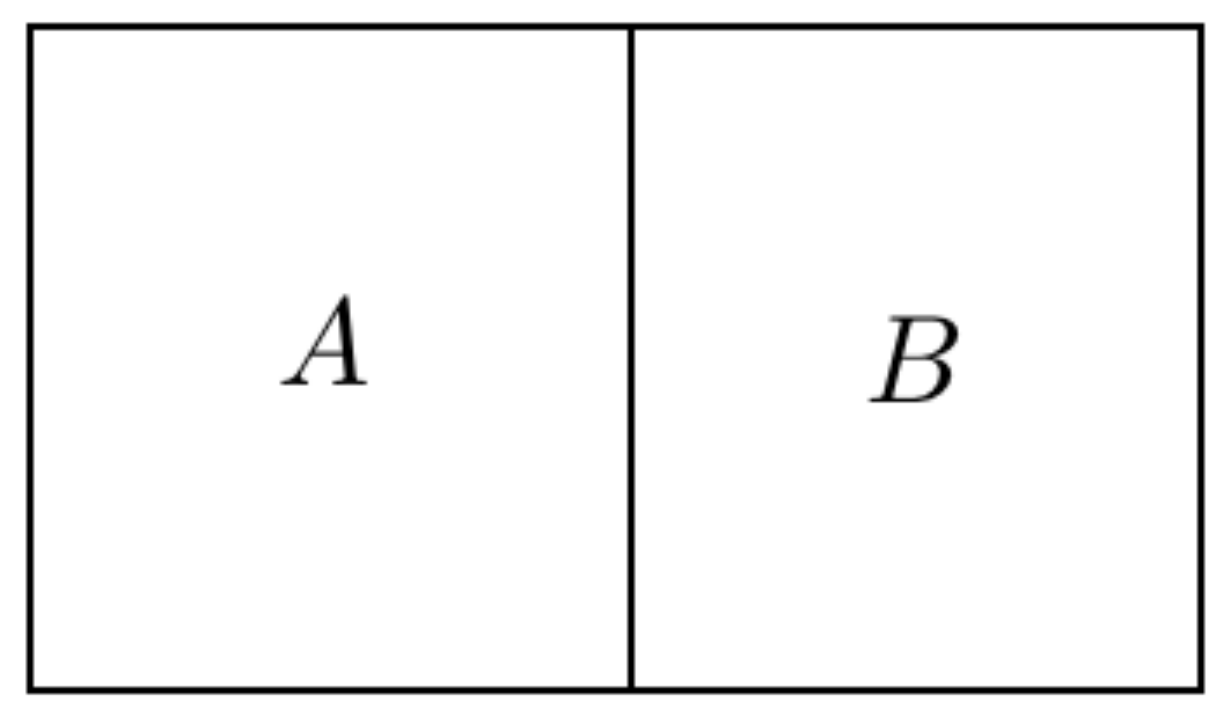}
  \caption{Setup for the homogeneity assumption.}
  \label{fig:AB}
\end{figure}

The next assumption relates the local terms of different
(overlapping) regions: we require that different regions have
approximately the same local terms over the segments where they overlap.

\begin{definition}[$\delta$-homogeneity]
  Let $\{f_{\partial A}\}_A$ be a family of quasi-local operators.
  We say that the family is \emph{$\delta$-homogeneous} if for every pair of rectangles $A$ and $B$
  arranged as in Fig.~\ref{fig:AB} and for every $Z\subset \partial A \setminus \partial B$,
  \begin{equation}
    \norm{f_Z^{\partial A} - f_Z^{\partial AB}} \le \delta^{AB}(\dist(Z,\partial B)),
  \end{equation}
  for some family $\{\delta^{A}(r)\}_{A}$ of decaying functions.
  If there exists a constant $r^*$ for which $\delta^A(r)=0$ for every $r>r^*$ and every $A$, we say that the family
  $\{f_{\partial A}\}_A$ is \emph{strictly homogeneous}.
\end{definition}
It is quite clear that strict homogeneity only makes sense in the case of strict locality.

In the following result, we will consider regions $ABC\subseteq\Lambda$ as in
Fig.~\ref{fig:ABC}, and we will furthermore subdivide region $z$ of $\partial B$ into two parts, which
we denote $x$ and $y$, as in Fig.~\ref{fig:bdregions-2}. We will assume that $B$ is $4\ell$ thick,
for some $\ell>0$, that moreover both region $x$ and $y$ are $\ell$
long.
\begin{figure}[hb]
  \centering
  \includegraphics[scale=0.35]{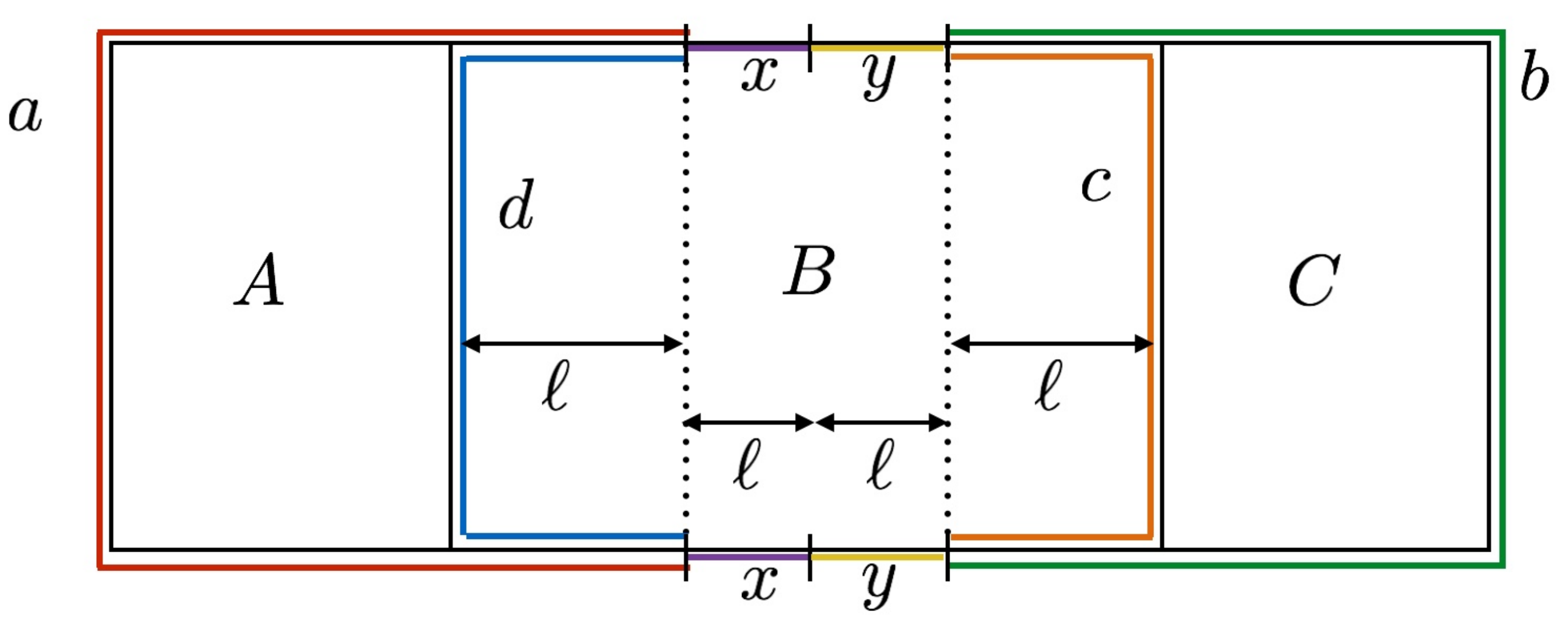}
  \caption{
    We further decompose region $z$ into $x$ and $y$ as compared with Fig.~\ref{fig:bdregions}.}
  \label{fig:bdregions-2}
\end{figure}

\begin{thm}\label{thm:strictlocality}
  Let us consider a family $\{Q_{\partial A}\}_{A\subset G}$ of local Hermitian operators with range
  $r$ and strength $J$, which is $\delta$-homogeneous, and let
  $\rho_{\partial A} = \exp(2 Q_{\partial A})$ the corresponding thermal state for each $A\subset G$
  (where the factor of 2 is added for convenience).  Then there exists an $\ell_0$ such that for a
  rectangular region $\Lambda$ satisfying the conditions described in Fig.~\ref{fig:ABC} with
  $\ell \ge \ell_0$, $\rho_{\partial \Lambda}$ is $\epsilon(\ell)$-approximately factorizable
  (equations \eqref{QF1} and \eqref{QF2}), where
  \begin{equation}
    \epsilon(\ell) = \eta\qty( c(r,J) \max\left\{\chi_\ell(2J), \sum_{l \ge \ell} \delta^\Lambda(l)
    \right\} ) ,
\end{equation}
where $\eta(t) = e^t-1 $, $c(r,J)$ depends only on the range $r$ and the strength $J$, and $\chi_\ell(2J)$ decays
faster than any exponential in $\ell$.
\end{thm}

\subsection{Analysis of 1D Gibbs states}\label{sec:Araki}
As a preliminary to the proof of Theorem~\ref{thm:strictlocality}, we will first recall some results
about 1D local Hamiltonians, and some useful theory on expansionals.  A very important set of tools
in our analysis of thermal states has been developed by H. Araki \cite{Araki} nearly a half a
century ago, with the purpose of proving the absence of phase transitions at finite temperatures for
1D local Hamiltonians. On the way to proving that result, he also showed that the evolution
associated to 1D local Hamiltonians is entire analytic, and that local observables preserve their
locality, up to small errors, under imaginary time evolution.

Here we state the results of Araki in a language that suits us, and refer back to the original paper
for the proofs.  We will make use of the following notation, for every pair of operators $A$ and
$B$:
\begin{equation}
  \Gamma^t_B(A) = e^{tB} A e^{-tB}.
\end{equation}
Note that if $[A,B] = 0$ then $\Gamma_B^t(A) =A$.

Consider the one dimensional local Hamiltonian $Q_\Sigma=\sum_{Z\subset \Sigma}q_Z$, where
$\Sigma \subset \bZ$ is a finite interval containing the origin, and the Hamiltonian terms $q_Z$ are
zero if the radius of $Z$ is greater than $r$.
For every integer $\ell$ let $Q_{[-\ell;\ell]}=\sum_{Z\subset [-\ell;\ell]}q_Z$ be the restriction to region $[-\ell;\ell]\subset\Sigma$; i.e. all
local Hamiltonian terms that are strictly inside of $[-\ell;\ell]$.
We will use the shorthand notation $\Gamma^t_\Sigma$ for $\Gamma^t_{Q_\Sigma}$ and $\Gamma^t_\ell$ for
$\Gamma^t_{Q_{[-\ell,\ell]}}$ for the rest of this section.

We can now state Arakis' theorem:
\begin{thm}[Araki \cite{Araki}]\label{thm:Araki}
Let $Q_\Sigma=\sum_{Z\subseteq \Sigma} q_Z$ be a local Hamiltonian on $\Sigma$ with
interaction length $r$ and strength $J$, and let $f$ be an observable with support on $[-n;n]$. Then
\bea \norm{\Gamma^t_\Sigma(f)-\Gamma^t_\ell(f)}&\leq& \chi_{\ell}(\tau)F_n(\tau)\norm{f}\label{eqn:Araki1}\\
 \norm{\Gamma^t_\Sigma(f)}&\leq& F_n(\tau) \norm{f},\label{eqn:Araki2}\eea
 where $\tau = 2tJ$ and the functions $\chi_\ell$ and $F_n$ can be bounded as
 \be F_n(x)\leq e^{(n-r+1)x+2\log(r)e^{xr}}, ~~~~~~~and
 ~~~~~~~\chi_\ell(x)\leq\frac{\left[2\log(r)e^{rx}\right]^{\lfloor \ell/r\rfloor+1}}{(\lfloor
   \ell/r\rfloor+1)!}.\label{eqn:Arakifuncs}\ee
\end{thm}
The bounds in Eqn.~(\ref{eqn:Arakifuncs}) only strictly holds for $x>c$ for some constant $c$ of order
one. For very small $t$, the bound on $F_n(x)$ and $\chi_\ell(x)$  take on a slightly different
functional form, that is mostly irrelevant for us now. The important thing to notice is that for any
fixed $n$ and $x$, $F_n(x)$ is bounded, and $\chi_\ell(x)$ is decaying faster than any exponential
in $\ell$, both uniformly in the size of $\Sigma$.

Araki's Theorem tells us that the imaginary time evolution of a local observable on a 1D line gets
mapped to a quasi-local observable with the same center for every time $t$, which can be well
approximated by evolution generated by restricted Hamiltonians.
Although for real time evolution, this statement holds in lattices of any dimension, and is more
widely known as a Lieb-Robison bound \cite{Lieb_1972}, it breaks down in 2 and higher dimensions for imaginary time.
In fact, Araki's Theorem shows that it is possible to extend the time evolution
$f \to e^{it Q_\Sigma}f e^{-it Q_\Sigma}$ in the limit where $\Sigma\to \bZ$, to an entire analytic function. This will not
be true in higher dimensions, and the resulting evolution will only be analytic in a strip centered
on the real line. An explicit counter-example was constructed in Ref. \cite{bouch}.

Note that, contrary to the Lieb-Robinson bound (which holds for real time evolution), where the
support of the evolved observable grows linearly with time, Araki's Theorem can only bound the
growth of the support with an exponential in time. Therefore the proof would not carry trough if we
defined quasi-local observables to have exponentially decaying tails.

We finally point out an easy corollary of Theorem \ref{thm:Araki} that we will need in Section \ref{sec:qfactPEPS}:
\begin{corollary}[Araki \cite{Araki}]\label{cor:Araki}
Let $Q_\Sigma=\sum_{Z\subseteq \Sigma} q_Z$ be a local Hamiltonian on $\Sigma$ with
interaction length $r$ and strength $J$, and let $f$ be a quasi-local observable with center at the origin. Then
\bea \norm{\Gamma^t_\Sigma(f)-\Gamma^t_\ell(f)} &\leq& \chi_{\ell}(\tau)G(\tau)\norm{f}\label{eqn:cor-Araki1}\\
 \norm{\Gamma^t_\Sigma(f)} &\leq& H(\tau) \norm{f},\label{eqn:cor-Araki2}\eea
for $\tau = 2tJ$ and for some analytic functions $G$ and $H$, which depend on $r$ but not on the
length of $\Sigma$.
\end{corollary}

It is immediately clear from Theorem \ref{thm:Araki} and Corollary \ref{cor:Araki} that if $f$ is
quasi-local then $\Gamma_\Sigma(f)$ is also quasi-local for every $t$, and its norm is bounded
uniformly in the size of $\Sigma$.

\subsection{Expansionals}
In order to construct the operators required to prove Theorem~\ref{thm:strictlocality}, we will need
the following object, which is known as a time-ordered exponential or expansional. For a detailed
account of its properties see \cite{fujiwara,Araki1973}.
\begin{definition}[Expansional \cite{fujiwara,Araki1973}]\label{def:expansional}
  Let $O: [0,1] \to \cB(\bC^D)$ be a continuous path of operators such that
  \begin{equation}\label{eq:exp-norm}
    \pathnorm{O} := \sup_{t \in [0,1]} \norm{O(t)} < \infty.
  \end{equation}
  The \emph{expansional} of $O(t)$, also known as the \emph{time-ordered exponential} of $O(t)$, is defined by
  \begin{equation}\label{eq:expansional}
    \OExp\left[ \int_0^1 \dd{t} O(t)\right] := \sum_{n=0}^\infty \int_0^1 \dd{t_1}
    \int_0^{t_1}\dd{t_2} \cdots \int_0^{t_{n-1}} \dd{t_n} O(t_1)\cdots O(t_n).
  \end{equation}
\end{definition}
Let us recall some useful properties of the expansional.
\begin{prop}\label{prop:expansional-properties}
  \mbox{}\\
  \begin{enumerate}
  \item If $[O(t_1),O(t_2)] = 0$ for all $t_1, t_2 \in [0,1]$, then $\OExp\left[ \int_0^1 \dd{t} O(t)
    \right]$ is equal to the usual exponential of the integral $\exp(\int_0^1 \dd{t} O(t))$.
    \item The norm of $\OExp \left[ \int_0^1 \dd{t} O(t) \right]$ is bounded by
      $\exp(\pathnorm{O})$, and moreover it holds that
    \begin{equation}\label{eq:expansional-norm}
      \norm{ \OExp\left[ \int_0^1 \dd{t} O(t) \right] - \1 } \le \exp(\pathnorm{O}) - 1.
    \end{equation}
    \item \cite[Proposition 15]{Araki1973} Let $O_1, \dots, O_n$ be operators. Then it holds that
    \begin{equation}\label{eq:iterated-expansional}
      e^{O_1}\cdots e^{O_n} = \OExp\left[ \int_0^1 \dd{t} \sum_{m=1}^n \Gamma_{O_1}^t \circ \cdots
        \circ \Gamma^t_{O_{m-1}}(O_m) \right].
    \end{equation}
  \end{enumerate}
\end{prop}
\begin{proof}[Proof of 2.]
  Since we have that
  \[\OExp\left[ \int_0^1 \dd{t} O(t) \right] = \1 + \sum_{n=1}^\infty
  \int_0^1 \dd{t_1}\int_0^{t_1}\dd{t_2} \cdots \int_0^{t_{n-1}} \dd{t_n} O(t_1)\cdots O(t_n),\]
  then we can bound the norm of $\OExp\left[ \int_0^1 \dd{t} O(t) \right] - \1$ as follows:
  \begin{eqnarray}
    \norm{ \OExp\left[ \int_0^1 \dd{t} O(t) \right] - \1 }
    &\le& \sum_{n=1}^\infty \int_0^1 \dd{t_1}\int_0^{t_1}\dd{t_2} \cdots \int_0^{t_{n-1}} \dd{t_n}
    \pathnorm{O}^n \nonumber\\
    &=& \sum_{n=1}^\infty \frac{\pathnorm{O}^n}{n!} = \exp(\pathnorm{O})-1.\nonumber
  \end{eqnarray}
\end{proof}

\subsection{The finite range Hamiltonian case}\label{sec:qfactPEPS}
We are now ready to prove Theorem \ref{thm:strictlocality}.
In order to show that the thermal states $\rho_{\partial ABC}$ are approximately factorizable, we will explicitly
construct operators
$\sigma_{\partial ABC}:=\Delta_{zb}\Delta_{az}$, $\sigma_{\partial B}:=\Omega_{zc}\Omega_{dz}$,
$\sigma_{\partial AB}:=\Omega_{zc}\Delta_{az}$, and $\sigma_{\partial BC}:=\Delta_{zb}\Omega_{dz}$
that satisfy Equations \eqref{QF1} and \eqref{QF2}.

To avoid excessive use of superscripts, we will denote the Hamiltonian associated to the
different regions as follows:
\bea
\rho_{\partial ABC}&=&e^{2Q_{axyb}}, {\rm with} ~~~ Q_{axyb}=\sum_{Z\subset axyb} q_Z\label{eqn:Q}\\
\rho_{\partial AB}&=&e^{2R_{axyc}}, {\rm with} ~~~  R_{axyc}=\sum_{Z\subset axyc} r_Z\label{eqn:R}\\
\rho_{\partial BC}&=&e^{2S_{dxyb}}, {\rm with} ~~~ S_{dxyb}=\sum_{Z\subset dxyb} s_Z\label{eqn:S}\\
\rho_{\partial B}&=&e^{2T_{dxyc}}, {\rm with} ~~~ T_{dxyc}=\sum_{Z\subset dxyc} t_Z.\label{eqn:T}
\eea

Given $\alpha$ one or more of the regions $a,b,c,d,x,y$, we will write $Q_\alpha=\sum_{Z\subseteq \alpha}Q_Z$ to mean all of the
local Hamiltonian terms of $Q_{axyb}$ that are strictly inside $\alpha\subseteq axyb$. The terms
intersecting two regions (say $\alpha$ and $\beta$) will be denoted $Q_{\partial \alpha\beta}$.
The same convention will be used for the other three Hamiltonians (Eqns.~(\ref{eqn:Q},\ref{eqn:R},\ref{eqn:S},\ref{eqn:T})).

We define the operators $\Delta$ and $\Omega$ as follows:
\bea \Delta_{axy}&:=& e^{Q_{ax}}e^{-Q_y}e^{Q_{axy}}\label{eqn:Deltaql}\\
\Omega_{xyc}&:=& e^{T_{xyc}}e^{-T_x}e^{T_{yc}}.\label{eqn:Deltaql2}\eea

We will use the properties of expansionals to show that, with these definitions of $\Delta$ and
$\Omega$, the boundary states are approximately factorizable.
We will first prove the following lemma, regarding $\sigma_{AB}$: the rest of the bounds of
Equations \eqref{QF1} and \eqref{QF2} can be proven in an analogous manner.

\begin{lemma}\label{lemma:strictlocality}
  With the definitions given above, we have that
  \begin{equation}
  \norm{\rho_{\partial AB}^{1/2} \sigma_{\partial AB}^{-1} \rho_{\partial AB}^{1/2} - \1} \le
  \eta(\epsilon(\ell)),
\end{equation}
where $\eta(t) = e^t-1$ and
\[ \epsilon(\ell) \le 2F_r(2J) \left[ 2r(1+H(2J)^2) + H(2J)G(2J) \right]\max\left\{  J\chi_{\ell}(2J),
      \sum_{l \ge \ell} \delta(l) \right\},
  \]
  and the functions $F_r$, $G$, $H$ and $\chi$ are given by Theorem \ref{thm:Araki} and Corollary \ref{cor:Araki}.
\end{lemma}
\begin{proof}
  To start with, rewrite the expression
  \bea \rho^{1/2}_{\partial AB}\sigma^{-1}_{\partial AB}\rho^{1/2}_{\partial AB}&=& e^{R_{axyc}} \Delta^{-1}_{axy}\Omega^{-1}_{xyc}e^{R_{axyb}}\nonumber\\
  &=& e^{R_{axyc}}e^{-Q_{axy}}e^{Q_y}e^{-Q_{ax}}e^{-T_{yc}}e^{T_x}e^{-T_{xyc}}e^{R_{axyc}}\nonumber\\
  &=& (e^{R_{axyc}}e^{-Q_{axy}}e^{Q_y}e^{-T_{yc}})(e^{-Q_{ax}}e^{T_x}e^{-T_{xyc}}e^{R_{axyc}})\nonumber\\
  &:=&O_LO_R,\nonumber\eea where we have used that $Q_{ax}$ and $T_{yc}$ commute since they are
  non-overlapping.  We will now show that there exists a decaying function $\epsilon(\ell)$,
  that we will specify later, such that
  \[
    \max(\norm{O_L -\1},\norm{O_R-\1}) \le \exp(\frac{\epsilon(\ell)}{2}) -1,
  \]
  which implies that
  \begin{equation}
    \norm{O_LO_R-\1}\leq \norm{O_L}~\norm{O_R-\1}+\norm{O_L-\1}\leq
    \exp(\epsilon(\ell)) - 1.
  \end{equation}
Now, in order to show that $O_L$ (or equivalently $O_R$) is close to the identity, we will apply
equation \eqref{eq:iterated-expansional} to the operator $O_L$ and obtain
\begin{multline*}
  O_L = e^{R_{axyc}}e^{-Q_{axy}}e^{Q_y}e^{-T_{yc}} = \\
  \OExp\left[ \int_0^1 \dd{t} R_{axyc} - \Gamma^t_{R_{axyc}}(Q_{axy}) +
    \Gamma^t_{R_{axyc}}\Gamma^{-t}_{Q_{axy}}(Q_y) -  \Gamma^t_{R_{axyc}}\Gamma^{-t}_{Q_{axy}}\Gamma^t_{Q_y}(T_{yc})  \right].
\end{multline*}

We will further decompose the r.h.s. of the previous equation as follows: we substitute $R_{axyc} =
\Gamma^t_{R_{axyc}}(R_{axyc})$ with $\Gamma^t_{R_{axyc}}(R_{axy} + R_{c} + R_{\partial yc})$, and
$\Gamma^t_{Q_y}(T_{yc})$ with $\Gamma^t_{Q_y}(T_y + T_{\partial yc}) + T_c$, since $[T_c,Q_y] = 0$.
The expression then reduces to
\begin{equation}\label{eqn:keyeq}
  O_L = \OExp\left[\int_0^1 \dd{t} \sum_{i=1}^5 X_i(t) \right],
\end{equation}
where
\begin{align*}
  X_1(t) &= \Gamma^t_{R_{axyc}}\left(R_{axy} - Q_{axy}\right),
  & X_2(t) &= \Gamma^t_{R_{axyc}}\left(R_c - T_c\right), \\
  X_3(t) &= \Gamma^t_{R_{axyc}}\Gamma_{Q_{axy}}^{-t}\Gamma_{Q_y}^t(Q_y -T_y),
  & X_4(t) &= \Gamma^t_{R_{axyc}}\Gamma_{Q_{axy}}^{-t}\Gamma_{Q_y}^t(R_{\partial yc}-T_{\partial yc}), \\
  X_5(t) &= \Gamma^t_{R_{axyc}}\left(R_{\partial yc}-\Gamma_{Q_{axy}}^{-t}\Gamma_{Q_y}^t(R_{\partial yc})\right). & &
\end{align*}
Let us denote $\epsilon(\ell) := 2\sum_i \pathnorm{X_i}$. Then by equation
\eqref{eq:expansional-norm} we have that
\[ \norm{O_L- \1} \le \exp(\pathnorm{\sum_i X_i}) -1 \le \exp(\frac{\epsilon(\ell)}{2}) -1. \]
Thus in order to bound $\norm{O_L-\1}$, it remains to show
that for each $i$ the norm $\pathnorm{X_i}$ is small.
The first term can be bounded as follows. Note
that \[ R_{axy}-Q_{axy}=\sum_{Z\subset axy} (r_Z-q_Z).\]
Then $(r_Z-q_Z)$ is zero if $Z$ has
radius larger than $r$. Then, from Theorem \ref{thm:Araki}, we get that $\Gamma^t_{R_{axyc}}$ acting
on a local operator is quasi-local in the sense of Def.~\ref{def:quasi-locality}, and that its norm is bounded
by a constant function $F$, hence
\begin{equation*}
  \norm{\Gamma^t_{R_{axyc}}(R_{axy}-Q_{axy})}
  \le F_r(2tJ) \sum_{Z\subset axy} \norm{r_{Z}-q_{Z}}
   \le F_r(2tJ) \sum_{Z\subset axy}\delta(d(Z,\partial C)) ,
\end{equation*}
where term on the r.h.s. is controlled by $\delta$-homogeneity. Since any $Z\subset axy$ is at least
at distance $\ell$ from $\partial C$, we obtain the bound:
\begin{equation}
  \pathnorm{X_1} \le r F_r(2J) \sum_{l > \ell} \delta(l).
\end{equation}
A similar argument works for $\pathnorm{X_2}$, which can be bounded in the same way. In order to
bound the norm of $X_3$ and $X_4$, we will have to add an extra step instead. We start from the same decomposition
\[
  X_3(t) = \sum_{Z\subset y} \Gamma^t_{R_{axyc}}\Gamma_{Q_{axy}}^{-t}\Gamma_{Q_y}^t(q_Z - t_Z),
\]
where the sum only runs over $Z$ with diameter smaller than the interaction length $r$. Each of the
three $\Gamma^t$ maps quasi-local operators to quasi-local operators, with a bound on the norm given
by equation \eqref{eqn:Araki2}, so that we get
\[ \norm{ \Gamma^t_{R_{axyc}}\Gamma_{Q_{axy}}^{-t}\Gamma_{Q_y}^t(q_Z - t_Z)} \le
  F_r(2tJ)H(2tJ)^2 \delta(d(Z,\partial B)), \]
so that again, we can bound
\[ \pathnorm{X_3} \le r F_r(2J)H(2J)^2 \sum_{l > \ell} \delta(l). \]
A similar analysis will work for $X_4$.

We now focus on $X_5$.  We expect $\Gamma_{Q_{axy}}^{-t}\Gamma_{Q_y}^t(R_{\partial yc})\approx
\Gamma_{Q_{y}}^{-t}\Gamma_{Q_y}^t(R_{\partial yc})=R_{\partial yc}$ when $y$ is large enough. Once
again, we use Theorem \ref{thm:Araki} and Corollary \ref{cor:Araki} to show that this indeed holds.
As for the other terms, we invoke the fact that quasi-local operators get mapped to quasi-local
operators under $\Gamma^t$. Furthermore,
we note that since $R_{\partial yc}$ is strictly local we get that $R_{\partial yc}-\Gamma_{Q_{axy}}^{-t}\Gamma_{Q_y}^t(R_{\partial yc})$ is quasi-local with center at $\partial yc$. This implies that
\be
\norm{\Gamma^t_{R_{axyc}}(R_{\partial yc}-\Gamma_{Q_{axy}}^{-t}\Gamma_{Q_y}^t(R_{\partial yc}))}
\leq H(2tJ)\norm{R_{\partial yc}-\Gamma_{Q_{axy}}^{-t}\Gamma_{Q_y}^t(R_{\partial yc})},
\ee
by Theorem \ref{thm:Araki}.
We now set out to show that $\norm{R_{\partial yc}-\Gamma_{Q_{axy}}^{-t}\Gamma_{Q_y}^t(R_{\partial yc})}$ is small:
\bea
\norm{R_{\partial yc}-\Gamma_{Q_{axy}}^{-t}\Gamma_{Q_y}^t(R_{\partial yc})} &
=& \norm{(\Gamma_{Q_{y}}^{-t}-\Gamma_{Q_{axy}}^{-t})\Gamma_{Q_y}^t(R_{\partial yc})}\nonumber\\
&\leq
&F_r(2tJ)G(2tJ) J\chi_\ell(2tJ).
\eea
Putting all of the bits together, we get that
\[ \epsilon(\ell) \le
  F_r(2J) \left[ 2r(1+H(2J)^2) + H(2J)G(2J) \right]\max\left\{  J\chi_{\ell}(2J),
      \sum_{l \ge \ell} \delta(l) \right\}.
 \]
\end{proof}

\subsection{Extension to quasi-local interactions}
In the proof of Lemma \ref{lemma:strictlocality}, Theorem \ref{thm:Araki} and Corollary
\ref{cor:Araki} played a crucial role. It is clear from the proof that the same approach will
generalize to a larger class of interactions as long as one is able to generalize equations
\eqref{eqn:cor-Araki1} and \eqref{eqn:cor-Araki2} to such Hamiltonians.

Abstractly, we can see Corollary \ref{cor:Araki} as a statement about
two classes of operators: on the one hand we have the local operators $\mathcal U_0$ and the
quasi-local ones $\mathcal U_1$. Then the result states that, for every $H\in \mathcal U_0$ and
$f\in \mathcal U_1$, the imaginary time evolution $\Gamma_H^t(f)$ still belongs to $\mathcal U_1$ for all
time $t$, and moreover that it can be well approximated by evolutions generated by a ``truncation''
of $H$ around the center of the support of $f$. In the proof of Lemma \ref{lemma:strictlocality} we
then apply the imaginary time evolution $\Gamma_H^t$ to the interaction terms of the Hamiltonians,
exploiting the fact that $\mathcal U_0 \subset \mathcal U_1$.

Extensions of Corollary \ref{cor:Araki} would then enlarge the classes $\mathcal U_0$ and
$\mathcal U_1$ for which these properties hold.  By using similar techniques to the original proof,
we believe that a first extension of Araki's theorem can be proven, to allow for quasi-local
Hamiltonian interactions (so that $\mathcal U_0 = \mathcal U_1$ are the quasi-local operators
defined in Definition \ref{def:quasi-locality}).
\begin{conjecture}[Quasi-local Araki]\label{cnj:qlAraki}
  Let $Q_\Sigma = \sum_{Z\subset \Sigma} q_Z$ be a quasi-local Hamiltonian on $\Sigma$ with strength
  $J$, and let $f$ be
  a quasi-local observable with center at the origin.
  Let $Q_{\ell} = \sum_{Z\subset [-\ell, \ell]} q_Z $. Then
  \begin{align}
    \norm{\Gamma^t_\Sigma(f)-\Gamma^t_\ell(f)} &\leq \mu_{\ell}(\tau)G(\tau)\norm{f}\label{eqn:Araki3}\\
    \norm{\Gamma^t_\Sigma(f)} &\leq H(\tau) \norm{f},\label{eqn:Araki4}
  \end{align}
  for $\tau=2Jt$ and for some analytic functions $G$ and $H$  which do not depend on the length of $\Sigma$.
\end{conjecture}
We do not provide a proof of this generalization of Theorem~\ref{thm:Araki}, but we observe that if
it holds, then the proof of Theorem \ref{thm:strictlocality} carries through verbatim of the case of
quasi-local interactions, obtaining the following:
\begin{thm}\label{thm:quasilocality}
  If Conjecture~\ref{cnj:qlAraki} holds, then thermal states of quasilocal Hamiltonians with
  strength $J$ and $\delta$-homogeneous are
  $\epsilon(\ell)$-approximately factorizable on regions
  $\{\partial ABC,\partial AB,\partial BC,\partial B\}$, with
  \begin{equation}
    \epsilon(\ell) = c(J)\left[ \mu_\ell(2J) + \sum_{l \ge \ell} \delta(l) \right]
  \end{equation}
  with $c(J)$ a positive constants independent of $\ell$.
\end{thm}

\subsection{Boundary Hamiltonians of PEPS}
We now comment on how to apply Theorem~\ref{thm:strictlocality} and Theorem~\ref{thm:quasilocality} in order to show that the boundary state of a PEPS is
approximately factorizabile. If the PEPS is injective on region $A$, then the boundary state $\rho_{\partial A}$ will be full
rank and therefore can be written as the Gibbs state of some Hermitian operator, which we will call
the \emph{boundary Hamiltonian}.
\begin{definition}\label{def:boundary-Hamiltonian}
  If a PEPS is injective on a region $A$, then the \emph{boundary Hamiltonian} is given by
  \begin{equation}
    Q_{\partial A} = \frac{1}{2} \log(\rho_{\partial A}),
  \end{equation}
  where $\rho_{\partial A}$ is the boundary state of the PEPS.
\end{definition}

To cover the case of $G$-injective or MPO-injective PEPS, one faces the problem that, in such cases,
boundary states are no longer full rank and hence cannot be Gibbs states of Hamiltonians. What is
then the structure that is expected to hold in the boundary state of a gapped $G$-injective or
MPO-injective PEPS? Numerical evidence from \cite{schuch2013topological} and analytical evidence
from \cite{cirac2017} suggests that the boundary states are, in that case, of the form $J_{\partial
  A} e^{2Q_{\partial A}}$, where $J_{\partial A}$ is the MPO projector on the boundary of region $A$
and the {\it boundary Hamiltonian} $Q_{\partial A}$ is  (quasi-)local, (quasi-)homogeneous and its
constituent interactions all commute with $J_{\partial A}$.

We can then modify Definition \ref{def:boundary-Hamiltonian} as follows.
\begin{definition}\label{def:mpo-boundary-Hamiltonian}
  If a PEPS is MPO-injective on a region $A$, and $J_{\partial A}$ is the projector on the kernel of
  $\rho_{\partial A}$, then the \emph{boundary Hamiltonian} is given by
  \begin{equation}
    Q_{\partial A} = \frac{1}{2} \log(\rho_{\partial_A}) \in
    J_{\partial A} \cH_{\partial A} J_{\partial A};
  \end{equation}
  where the logarithm is understood to be restricted to the support of $\rho_{\partial A}$.
\end{definition}

Does the boundary Hamiltonian of a PEPS defined in this way satisfy the assumptions of
Theorem~\ref{thm:strictlocality}? While we do not have a satisfying answer except for the case of
isometric and G-isometric PEPS (for which the parent Hamiltonian is commuting), we comment on what
numerical evidence can tell us. The range of the boundary Hamiltonian has been investigated in
detail in Ref. \cite{ciracES} for the square lattice AKLT and the Ising PEPS models on a
cylinder. There, the authors numerically compared the boundary Hamiltonians on the cylinder to the
long range Heisenberg Hamiltonian:
\begin{equation}
  H=\sum_{\ell\geq1}\eta_\ell\sum_{j\in \lambda} S_j S_{j+\ell} +R,
\end{equation}
where $R$ is some unknown rest term. They extracted the values of $\eta_\ell$ for a depth two and
for an infinite depth cylinder, and found to very high accuracy that the terms $\eta_\ell$ decayed
exponentially with $\ell$ for $\ell>2$. The norm of $R$ was also shown to be very small. They
perform the same numerics for Ising PEPS in the non-critical regime, and also observe that the
boundary Hamiltonian shows some decay in the interaction strength. In the regime where the Ising
PEPS becomes non-ergodic, they find that the boundary state begins to resemble a mean field
Heisenberg model; i.e. $\eta_\ell=O(1)$ for all $\ell$. While it is not reasonable to expect a
finite range boundary Hamiltonian from these results, it seems quite challenging to distinguish a
faster than exponential decay from an exponential decay, meaning that even assuming that the
quasi-local version of Araki given in Conjecture~\ref{cnj:qlAraki} holds, one could not easily apply
Theorem~\ref{thm:quasilocality}.

Unfortunately, the authors of Ref. \cite{ciracES} did not make any specific statements about the
homogeneity of the boundary state, however there is some evidence in their numerics to support our
assumptions. In \cite[Fig. 8a]{ciracES} the magnitude of the $\eta_\ell$ term was plotted for different cylinder
diameters. The values of $\eta_\ell$ are essentially independent of the cylinder diameter as long as
$\ell<L/2$, where $L$ is the cylinder diameter. This suggests, at least in the translationally
invariant case, that the boundary Hamiltonian has a universal character, and only the (exponentially
suppressed) very long range contributions are perturbed with changes at a long distance.

Beyond the numerical evidence above, locality of the boundary state has been shown to hold
analytically for models of non-interacting free-fermions \cite{fidkowski2010}, and for certain
conformal field theories \cite{lou2011}. However, caution must be taken, since systems with a chiral
symmetry give rise to critical boundary states.

\section{Conclusion}

We have proved a fundamental theorem relating boundary states of two dimensional injective PEPS to
the bulk gap of the parent Hamiltonian. Our work raises a lot of further questions for our
understanding and analysis of bulk-boundary correspondences in many body systems as well as for
relating static to dynamic properties in physical systems, which we discuss below.

\paragraph{Is the approximate factorization condition necessary?}
The most pressing question is perhaps to know if the assumptions on the boundary states made in
Theorem~\ref{thm:main}, which are sufficient to prove a spectral gap in the bulk, are also
necessary. There are good reasons to believe that this is the case. It was recently shown
\cite{SchwarzPEPS} that a uniform bulk gap of the parent Hamiltonian which is similar (although not
identical) to Definition \ref{def:spectral-gap} implies that the PEPS satisfies local
indistinguishability. Local indistinguishability, and its topological variant, LTQO \cite{LTQOPEPS},
imply that local observables can be evaluated accurately by only contracting a finite ring of
tensors around the observable.
We believe (although we do not have a proof at present) that the property of LTQO should allow us to
show that shielded regions of the boundary states satisfy the decay of mutual information bound that
has recently been shown to by equivalent to the existence of local recovery maps
\cite{fawzi2014}. This in turn implies that the boundary state is close to a local (although not
necessarily bounded) Gibbs state \cite{KatoBrandao16}.  Given this insight, as well as the numerical
evidence from \cite{ciracES}, we conjecture the following:
\begin{conjecture}
  If for any rectangular region $A\subset\Lambda$, $H_A$ is gapped, then the boundary states of $A$
  are close to a Gibbs state of a 1D Hamiltonian with exponentially decaying interactions,
  and they are approximately factorizable.
\end{conjecture}

More abstractly, one might ask whether there are further correspondences between bulk and
boundary properties, including conventional symmetries. One interesting direction to look into is
whether the effective temperature of the boundary Hamiltonian is related to the correlation length
of bulk observables as predicted by Poilblanc \cite{poilblanc2010}.

\paragraph{A canonical form for PEPS?} In every PEPS there is a gauge degree of freedom in its
defining tensor, in the sense that if we multiply the left and right virtual levels by $Y$ and
$Y^{-1}$ respectively, this action gets canceled in the tensor contraction that defines the
PEPS. The same happens for the top and bottom virtual levels when multiplying by $Z$ and $Z^{-1}$
respectively. Moreover, by considering non-translation invariant tensor networks it is possible to
change the choice of gauge matrices at each edge. While this operation does not change the physical
state represented by the PEPS, it will transform the boundary state via a product of congruence (but
not similarity) transformations, in the sense that the boundary state $\rho_{\partial A}$ of a
region $A$ will be mapped to
$(X_1^\dag\otimes \dots \otimes X^\dag_{\abs{\partial A}}) \rho_{\partial A}
(X_1 \otimes\dots \otimes X_{\abs{\partial A}})$, where $\{X_i\}_i$ are invertible matrices.
It has been proved in \cite{perez2010characterizing} that for injective PEPS
this is the only freedom in the PEPS tensor. It is not clear however how to choose the {\it best}
gauge matrices for a given PEPS.
In 1D, the canonical form \cite{MPSrep} defined in Lemma
\ref{lemma:mps-state-convergence} gives a way to fix the gauge which implies the required
approximate factorization of the boundary state. Based on this, one could then define the canonical
gauge in 2D exactly as the one needed to have an approximate factorization of the boundary state (in
case such factorization exists). Note that since the gauge transformation can potentially change both the
eigenvalues and the eigenvectors of the boundary state, it is not clear that having the approximate
factorization property for a given choice of gauge implies the same for other choices of gauges
(given that square roots of the boundary states appear in \eqref{QF1} and \eqref{QF2}).

\subsection*{Acknowledgments}
We thank Albert Werner and Wojciech De Roeck for fruitful discussions.
M.\,J.\,K.\, was supported by the VILLUM FONDEN Young Investigator Program.
A.\,L.\, acknowledges financial support from the European Research Council (ERC Grant Agreement no
337603), the Danish Council for Independent Research (Sapere Aude), VILLUM FONDEN via the QMATH
Centre of Excellence (Grant No. 10059), the Walter Burke Institute for Theoretical Physics in the
form of the Sherman Fairchild Fellowship as well as support from the Institute for Quantum
Information and Matter (IQIM), an NSF Physics Frontiers Center (NFS Grant PHY-1733907).
D.\,P.\,G.\, acknowledges support from MINECO (grant MTM2014-54240-P), Comunidad de Madrid (grant
QUITEMAD+-CM, ref. S2013/ICE-2801), and Severo Ochoa project SEV-2015-556. This project has received
funding from the European Research Council (ERC) under the European Union's Horizon 2020 research
and innovation programme (grant agreement No 648913).

\bibliographystyle{unsrt}
\bibliography{PEPSGap}

\end{document}